\documentclass[preprint,12pt,authoryear]{elsarticle}

\usepackage[a4paper,left=0.95in,right=0.95in,top=1.0in,bottom=1.0in]{geometry}

\usepackage[T1]{fontenc}
\usepackage[utf8]{inputenc}
\usepackage{amsmath,amssymb,amsthm,bm,mathtools}
\usepackage{graphicx}
\usepackage{booktabs}
\usepackage{longtable}
\usepackage{array}
\usepackage{multirow}
\usepackage{tabularx}
\usepackage{threeparttable}
\usepackage{caption}
\usepackage{subcaption}
\usepackage{algorithm}
\usepackage{algpseudocode}
\usepackage{enumitem}
\usepackage{setspace}
\usepackage[protrusion=true,expansion=false]{microtype}
\usepackage{xcolor}
\usepackage{url}
\usepackage[colorlinks=true,linkcolor=blue,citecolor=blue,urlcolor=blue]{hyperref}
\usepackage[nameinlink,noabbrev]{cleveref}
\usepackage{placeins}
\usepackage{needspace}

\graphicspath{{figures/}}
\allowdisplaybreaks
\newtheorem{assumption}{Assumption}
\newtheorem{proposition}{Proposition}
\newtheorem{theorem}{Theorem}
\newtheorem{corollary}{Corollary}
\newtheorem{lemma}{Lemma}

\crefname{assumption}{Assumption}{Assumptions}
\crefname{corollary}{Corollary}{Corollaries}
\Crefname{corollary}{Corollary}{Corollaries}
\crefformat{appendixref}{#2#1#3}
\Crefformat{appendixref}{#2#1#3}
\crefrangeformat{appendixref}{#3#1#4--#5#2#6}
\Crefrangeformat{appendixref}{#3#1#4--#5#2#6}
\crefmultiformat{appendixref}
  {#2#1#3}{ and #2#1#3}{, #2#1#3}{, and #2#1#3}
\Crefmultiformat{appendixref}
  {#2#1#3}{ and #2#1#3}{, #2#1#3}{, and #2#1#3}

\newcommand{\ind}{\mathbf{1}}

\newcommand{\E}{\mathbb{E}}

\newcommand{\loss}{\ell}
\newcolumntype{Y}{>{\raggedleft\arraybackslash}X}

\begin{document}

\begin{frontmatter}

\title{Efficient conformal prediction intervals for time series: Online PID-Expert aggregation}

\author[buaa]{Guodong Liu}
\ead{guodongliu@buaa.edu.cn}
\author[buaa]{Yanfei Kang}
\ead{yanfeikang@buaa.edu.cn}
\author[carbon]{Ren Miao}
\ead{miaoren@d-shuji.cn}
\author[cas]{Xiaoqian Wang\corref{cor}}
\ead{xiaoqian.wang@amss.ac.cn}
\cortext[cor]{Corresponding author}

\affiliation[buaa]{organization={School of Economics and Management, Beihang University}, city={Beijing 100191}, country={China}}
\affiliation[carbon]{organization={Beijing Shuji Intelligent Technology Co., Ltd.}, city={Beijing 100176}, country={China}}
\affiliation[cas]{organization={Center for Forecasting Science, Academy of Mathematics and Systems Science, Chinese Academy of Sciences}, city={Beijing 100190}, country={China}}

\begin{abstract}

For a given point forecaster, proportional--integral--derivative (PID) calibration configurations can attain similar overall coverage yet produce different interval widths. We introduce PID-Expert, an online aggregation method for improving the efficiency of conformal prediction intervals for time series. PID-Expert combines thresholds from a fixed library of PID calibrators, each evolving under its own coverage feedback. Aggregation weights depend on normalized interval width and miscoverage, while a shared multiplier adapts the miscoverage penalty using feedback from the reported interval. We establish local regret bounds for weighted expert losses under the realized multiplier sequence and, separately, a pathwise upper bound on time-averaged aggregate miscoverage, with control in expectation and almost surely under a stability condition. Across four simulation settings and two real-data applications, PID-Expert produces narrower mean intervals than a prespecified Conformal PID benchmark while keeping overall empirical coverage close to the nominal level. Compared with expert selection and equal averaging, aggregation generally provides a more balanced coverage--width trade-off across forecasting settings. Rolling analyses further reveal local coverage--efficiency trade-offs, particularly following abrupt distributional shifts. Overall, PID-Expert reduces reliance on a single PID configuration while improving interval efficiency.

\end{abstract}


\end{frontmatter}

\section{Introduction}

Online prediction intervals need to keep miscoverage near a target rate while avoiding unnecessary width. This balance between calibration and sharpness is central to probabilistic forecasting \citep{gneiting2007,gneiting2014} and interval forecasting practice \citep{chatfield1993,hyndman2021,petropoulos2022}. For a given point forecaster, calibration rules with similar overall coverage can still produce intervals of different widths. Online conformal methods adjust calibration through realized coverage errors as new observations arrive \citep{gibbs2021,angelopoulos2023pid}. Even when these adjustments bring overall coverage close to the target, the choice of calibration rule can leave room to improve interval efficiency.

Conformal proportional--integral--derivative (PID) control makes this choice concrete \citep{angelopoulos2023pid}. It updates score thresholds through recent coverage errors, accumulated coverage deviations, and forecasts of future nonconformity scores. The proportional gain, integral gain, and saturation constant determine how these feedback signals affect the threshold path. Faster responses can produce more variable intervals, while slower or more conservative responses can require wider intervals. Coverage feedback adjusts an individual controller but does not directly compare its interval efficiency with that of other configurations. Thus, even with a PID calibrator in place, a practical choice remains: which feedback configuration avoids unnecessary width while keeping coverage close to the target?

To assess how much interval efficiency depends on this choice, we compare PID configurations under a common forecasting setup. In the nonlinear simulation in \Cref{sec:experiments}, 343 PID controllers share the same point forecasts, nonconformity scores, and scorecaster (\Cref{fig:nonlinear-expert-scatter}). Their empirical coverage ranges from $0.901$ to $0.909$ at a target of $0.9$, yet their mean interval widths differ. The prespecified Conformal PID controller has mean width $0.360$, compared with $0.352$ for PID-Expert and $0.348$ for the narrowest library member selected retrospectively among those with empirical coverage of at least $0.9$. These differences motivate combining PID configurations online, using their observed performance to guide the construction of prediction intervals.

The closest existing approaches address related choices at different stages of forecasting and calibration. Adaptive conformal inference (ACI), introduced by \citet{gibbs2021}, adjusts the calibration level in response to observed coverage errors. AgACI, as originally proposed by \citet{zaffran2022}, aggregates interval bounds from ACI predictors with different step sizes using Bernstein online aggregation \citep{wintenberger2017}, while DtACI \citep{gibbs2024} adapts their influence through their calibration levels. These methods reduce reliance on a single ACI learning rate. Forecast combinations address uncertainty at the predictive-model stage \citep{Petropoulos2018,Wang2023}. Model-selection approaches also adapt across model-specific conformal predictors \citep{hajihashemi2024ensemble}, and Bellman conformal inference optimizes interval length using families of nominal prediction intervals \citep{yang2024ts}. Our focus is the efficiency left by PID calibration configurations for a shared point forecaster. We combine the threshold paths generated by distinct feedback histories and use coverage of the final reported interval to adjust how these paths are evaluated.

We propose \emph{PID-Expert aggregation} with two connected feedback layers. At the expert layer, each PID controller maintains its own coverage history and generates a candidate threshold. Keeping these states separate preserves the different calibration paths available for comparison. At the aggregation layer, weights combine the candidate thresholds and are updated using normalized interval width and expert miscoverage. Coverage of the reported interval updates a common multiplier that sets the relative penalty on miscoverage in these expert losses. This additional feedback is needed because coverage of an averaged threshold is not the weighted average of the experts' coverage indicators. The multiplier therefore brings the reported interval's coverage into the comparison of calibration paths.

The paper makes three contributions. First, PID-Expert aggregation turns reliance on one prespecified PID configuration into an online comparison of independently evolving calibration paths, with an explicit width--miscoverage loss informed by aggregate coverage. Second, the analysis separates the two questions created by this design: local regret bounds compare weighted expert losses with a fixed library member under the realized multiplier sequence, while a separate pathwise relation links aggregate miscoverage to multiplier growth. A stability condition then gives time-averaged upper control in expectation and almost surely; the regret bound also motivates reference aggregation parameters. Third, four simulation settings and two real-data applications evaluate PID-Expert aggregation against the prespecified Conformal PID benchmark, DtACI, and AgACI, together with expert selection and simple averaging as structured benchmarks for isolating the role of online aggregation. Across the reported forecasting configurations, PID-Expert aggregation produces narrower mean intervals than the prespecified PID benchmark while maintaining empirical coverage close to the target. The comparison with selection and averaging shows that the gains are not explained simply by choosing a favorable expert or averaging expert outputs, while rolling-window analyses further characterize the associated local coverage--efficiency tradeoffs.

The remainder of the paper is organized as follows. \Cref{sec:literature} reviews the relevant literature and introduces the online prediction setup and conformal PID calibration. \Cref{sec:method} presents PID-Expert aggregation. \Cref{sec:parameter} develops the regret and coverage analysis. \Cref{sec:experiments} evaluates the method through simulations and real-data applications. \Cref{sec:conclusion} discusses the findings and concludes. Proofs and additional implementation details are provided in the appendices.

\section{Background and Problem Setup}
\label{sec:literature}

This section places the calibration problem in the time-series literature, defines the online prediction setting, and introduces the PID states used by the expert library.

\subsection{Conformal prediction for time series}
\label{sec:time-series-background}

Conformal prediction constructs prediction sets by calibrating nonconformity scores, which measure disagreement between predictions and observations \citep{vovk2005,shafer2008}. Split conformal prediction separates model fitting from calibration and selects a score quantile from a calibration sample, obtaining finite-sample marginal coverage under exchangeability \citep{papadopoulos2002,lei2018,angelopoulos2023gentle}.

Temporal dependence and structural change motivate approaches beyond exchangeable calibration. In the forecasting literature, \citet{KOO2015166} study structural-break models under misspecification and their implications for forecasting, while \citet{Koo17022026} distinguish changes in factor variances from changes in loadings in macroeconomic factor models. In conformal prediction, weighted calibration addresses covariate shift and other departures from exchangeability \citep{tibshirani2019,barber2023}, while permutation-based methods account for dependent data \citep{chernozhukov2018}. \citet{xu2021} combine bootstrap ensembles with sequential residual calibration, and \citet{xu2023} use temporal dependence in scores to estimate conditional score quantiles. Their guarantees depend on the respective assumptions about the score process and estimation accuracy.

Other extensions broaden the prediction target. Conformal methods have been developed for multiple forecast horizons \citep{stankeviciute2021,wang2024cpts}, multistep multivariate prediction \citep{schlembach2025}, and functional time series with applications to age-specific mortality \citep{shang2026conformal}. These studies extend calibration to prediction tasks with different temporal and response structures.

Online feedback provides a further approach: ACI uses realized coverage errors to update subsequent calibration levels \citep{gibbs2021}. Building on ACI, \citet{zaffran2022} aggregate intervals from predictors with different step sizes, while \citet{gibbs2024} dynamically adapt the choice of step size. Such updates can also exploit predictable structure. \citet{sun2025} combine latent-state prediction with adaptive calibration to address change points. \citet{areces2025optimization} formulate online conformal prediction through online optimization, and \citet{hu2026distribution} incorporate estimated score distributions into optimistic threshold updates. We study one-step symmetric intervals, holding the point forecasts and score information fixed to compare the efficiency of different PID calibration paths.

\subsection{Online prediction setup}
\label{sec:online-setup}

Consider responses $y_t\in\mathbb R$ observed sequentially, with any covariates available at the forecast origin. Before observing $y_t$, we construct a one-step-ahead point forecast $\widehat y_{t|t-1}$ and choose a nonnegative threshold $q_t$. For a fixed target miscoverage rate $0<\alpha<1$, the reported interval is
\begin{equation*}
    \widehat C_{t,1-\alpha}
    =[\widehat y_{t|t-1}-q_t,\widehat y_{t|t-1}+q_t].
\end{equation*}
After observing the response, we obtain the nonconformity score and miscoverage indicator
\begin{equation*}
    s_t=|y_t-\widehat y_{t|t-1}|,
    \qquad
    err_t=\ind\{y_t\notin\widehat C_{t,1-\alpha}\}
           =\ind\{s_t>q_t\}.
\end{equation*}
The interval width is $2q_t$. Increasing the threshold can reduce miscoverage at the cost of width, so interval efficiency must be assessed alongside coverage.

We allow temporal dependence and distributional change and study calibration through sequential feedback. The objective is to keep the time-averaged miscoverage $T^{-1}\sum_{t=1}^T err_t$ near $\alpha$ while avoiding unnecessary width. This criterion differs from coverage at each forecast origin. The theory in \Cref{sec:parameter} gives upper control of time-averaged miscoverage under its stated conditions.

For the quantiles used below, write
\begin{equation*}
    Q_\tau(P)=\inf\{r\in\mathbb R:P((-\infty,r])\ge\tau\},
\end{equation*}
where $P$ is a probability measure, and let $\delta_s$ denote a unit point mass at $s$.

\subsection{Conformal PID calibration}
\label{sec:pid-calibration}

Conformal PID updates the threshold directly on the score scale \citep{angelopoulos2023pid,wang2024cpts}. Its components respond to recent miscoverage, accumulated coverage deviations, and predictable changes in future scores. These components make a controller's configuration relevant to both its coverage response and its interval width.

\Needspace{5\baselineskip}
After observing the score at time $t$, the controller updates the nonnegative threshold as
\begin{equation}
    q_{t+1}
    =
    \left[
        q_{t+1}^P
        +
        K_I
        \tan\!\left(
            \frac{I_t\log t}
                 {tC_{\mathrm{sat}}}
        \right)
        +
        \widehat s_{t+1\mid t}
    \right]_{+},
    \label{eq:pid-update}
\end{equation}
where $[z]_+=\max\{z,0\}$.

The first component is the proportional, or quantile-tracking, state
\begin{equation*}
    q_{t+1}^P
    =
    q_t^P+\eta R_t(err_t-\alpha),
\end{equation*}
where $\eta>0$ is the proportional learning rate that controls the responsiveness to the latest coverage error. For $t\ge2$, we define the recent score scale as
\begin{equation*}
    R_t
    =
    \max_{r\in\mathcal W_t^P}s_r
    -
    \min_{r\in\mathcal W_t^P}s_r,
    \qquad
    \mathcal W_t^P
    =
    \{\max(1,t-L_P),\ldots,t-1\}.
\end{equation*}
This range is based on a lookback window of length $L_P$. In our implementation, we set $R_1=1$ when no previous score is available.

The integral component responds to accumulated deviations from the target miscoverage rate through
\begin{equation*}
    I_t
    =
    \sum_{r=1}^{t}(err_r-\alpha).
\end{equation*}
Positive values of $I_t$ indicate cumulative undercoverage, whereas negative values indicate cumulative overcoverage. The integral gain $K_I\ge0$ determines the magnitude of the correction on the score scale, while the saturation constant $C_{\mathrm{sat}}>0$ governs how rapidly the correction steepens as the cumulative error grows. Within the unsaturated range, smaller values produce stronger corrections for the same error at a given time. In our implementation, the tangent argument is clipped inside $(-\pi/2,\pi/2)$ to keep the update finite.

The final component, $\widehat s_{t+1\mid t}$, forecasts the next nonconformity score using information available through round $t$. This scorecaster provides the anticipatory component of the controller by exploiting predictable temporal structure in the score sequence beyond that captured by coverage feedback alone.

Together, the three components allow the controller to respond to recent coverage errors, correct persistent deviations from the target, and anticipate changes in the scale of forecast uncertainty. Holding the scorecaster and lookback rule fixed, its behavior depends on the controller configuration
\begin{equation*}
    \theta=(\eta,K_I,C_{\mathrm{sat}}).
\end{equation*}
Different controller configurations can generate markedly different adjustment dynamics following a regime change. As illustrated schematically in \Cref{schematic}, an aggressive configuration may restore coverage rapidly but induce oscillations in both coverage and interval width, whereas a slower configuration can exhibit prolonged undercoverage as interval width adjusts only gradually. A more conservative configuration may maintain coverage above the target at the cost of persistently wider intervals. By contrast, a balanced configuration can deliver a more stable adjustment in both coverage and interval width. These contrasting responses motivate dynamically combining heterogeneous PID controllers rather than relying on a single prespecified configuration.

\begin{figure}[!htbp]
\centering
\includegraphics[width=0.6\textwidth]{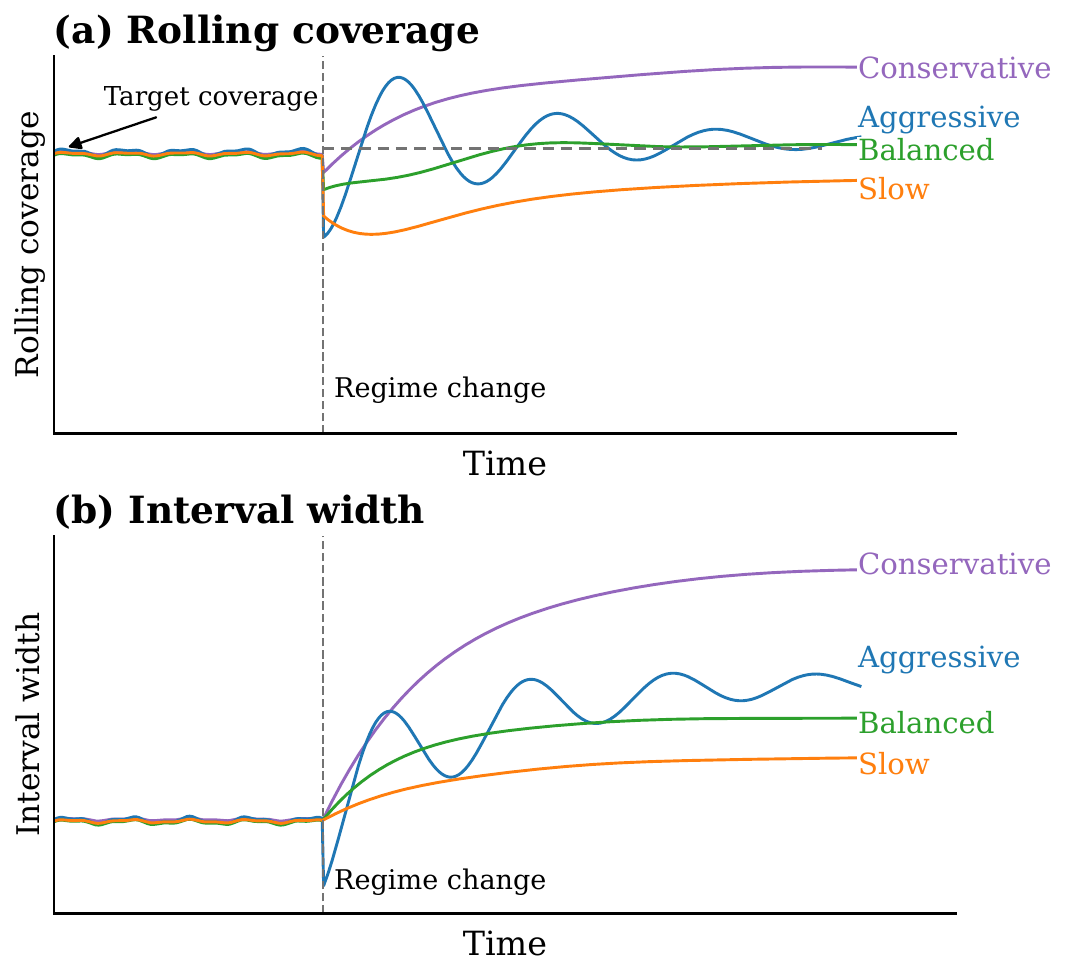}
\caption{Schematic illustration of coverage and interval width dynamics under different PID controller configurations following a regime change.}
\label{schematic}
\end{figure}
\FloatBarrier
\section{PID-Expert aggregation}
\label{sec:method}

\subsection{Overview}

Our objective is to achieve a favorable coverage--efficiency tradeoff under evolving data distributions by constructing narrow prediction intervals while maintaining coverage close to the target level. To this end, we propose \emph{PID-Expert aggregation}, an online method that adaptively combines candidate thresholds from a fixed library of PID controllers with heterogeneous parameter configurations. All experts share the same point forecasts and sequence of nonconformity scores, but each evolves according to its own history of coverage feedback. The aggregation mechanism is designed to account jointly for coverage performance and interval efficiency when adapting the relative influence of the experts.

\begin{figure}[!htbp]
    \centering
    \includegraphics[width=\textwidth]{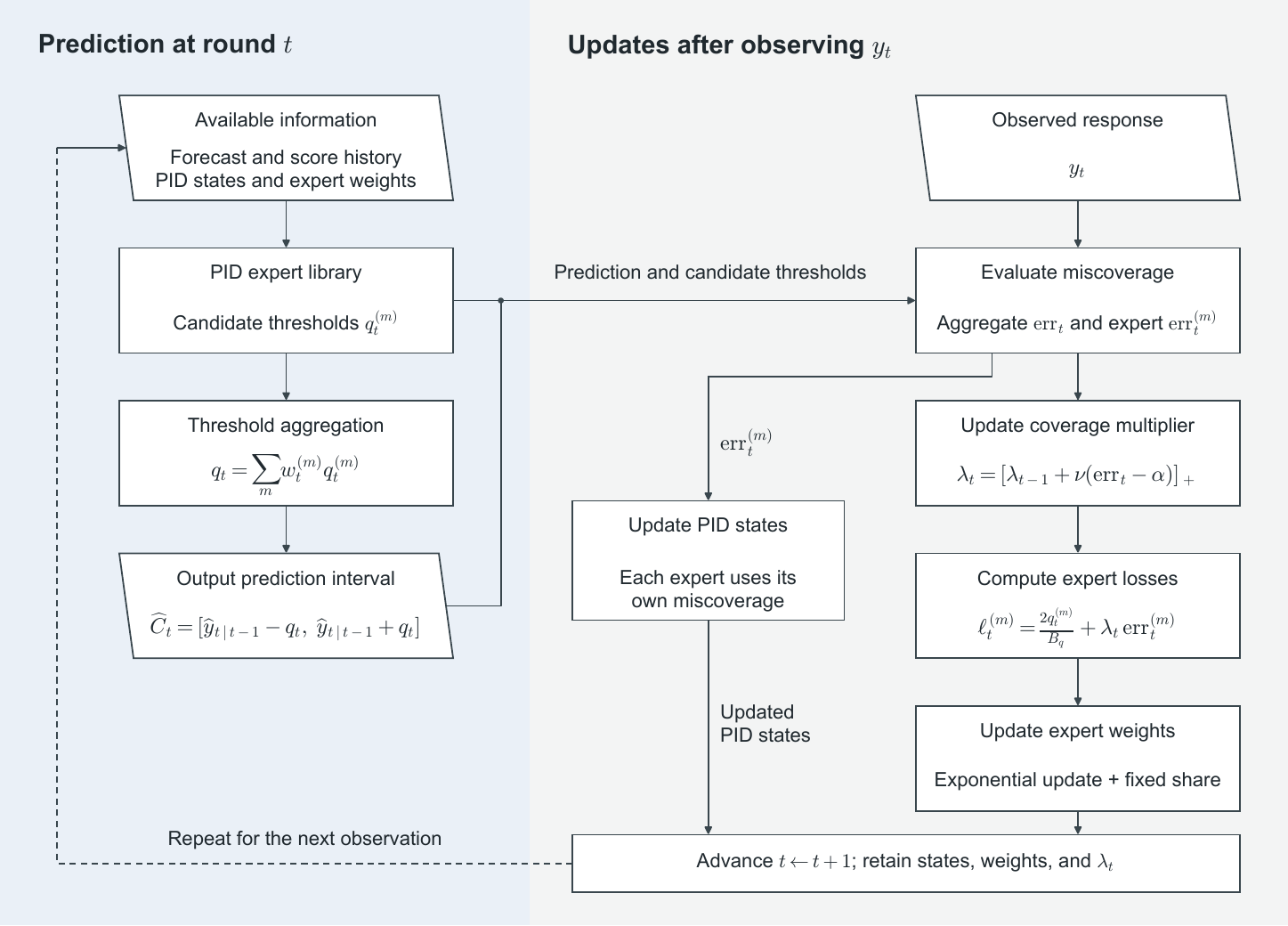}
    \caption{Framework of PID-Expert aggregation for sequential forecasting.}
    \label{fig:framework}
\end{figure}

As illustrated in \Cref{fig:framework}, the method proceeds recursively through prediction and updating based on feedback. At each round, the candidate PID thresholds are combined using the current expert weights to form the reported prediction interval. Once the response is observed, aggregate coverage feedback adjusts the emphasis placed on miscoverage, while each expert's realized coverage and interval width inform the subsequent weight update. In parallel, each expert updates its own PID state using its individual coverage feedback. The resulting states and weights are then carried forward to the next round, allowing the ensemble to adapt its coverage--efficiency balance over time while preserving the distinct calibration dynamics of the individual experts.

\subsection{Expert library}

We construct a fixed library of PID controllers with heterogeneous configurations to provide a range of threshold responses to changing distributions of nonconformity scores. Expert $m$ is specified by the parameter vector
\begin{equation*}
    \theta_m=(\eta_m,K_{I,m},C_{\mathrm{sat},m}),
    \qquad m=1,\ldots,M.
\end{equation*}
All experts share the same point forecast, sequence of nonconformity scores, predictable score range $R_t$, and scorecaster contribution. Each expert maintains its own proportional state $q_t^{P,(m)}$, cumulative coverage error $I_t^{(m)}$, and candidate threshold $q_t^{(m)}$.

After observing the response, we compute each expert's miscoverage indicator
\begin{equation*}
    err_t^{(m)}=\ind\{s_t>q_t^{(m)}\}.
\end{equation*}
Each controller follows the PID formulation in \Cref{eq:pid-update}, using its own parameters and miscoverage feedback. Specifically, its threshold is updated as
\begin{equation*}
    q_{t+1}^{(m)}
    =\left[
      q_{t+1}^{P,(m)}
      +K_{I,m}\tan\!\left\{
        \frac{I_t^{(m)}\log t}{tC_{\mathrm{sat},m}}
      \right\}
      +\widehat s_{t+1\mid t}
    \right]_+.
\end{equation*}
Thus, the experts retain distinct threshold trajectories driven by their own coverage histories. We aggregate these trajectories rather than average the PID parameters and run a single controller. Averaging parameters would collapse the separate coverage histories and would generally produce a different threshold path, since each expert's miscoverage indicator depends on its own threshold. The aggregation weights leave the individual PID updates unchanged.

\subsection{Coverage--efficiency loss and online aggregation}

Coverage alone cannot distinguish an efficient interval from an unnecessarily wide interval that also covers the response. We therefore compare experts using both width and miscoverage. After observing the response at round $t$, expert $m$ receives the loss
\begin{equation}
    \loss_t^{(m)}
    =\frac{2q_t^{(m)}}{B_q}+\lambda_t err_t^{(m)}.
    \label{eq:binary-meta-loss}
\end{equation}
The width term favors narrower intervals, while the miscoverage term penalizes miscoverage according to the current multiplier $\lambda_t$. We use an initial window of $L_0$ nonconformity scores during an initial period before online evaluation and set the nominal normalization scale to
\begin{equation}
    B_q=2Q_{1-\alpha}\!\left(\frac{1}{L_0}\sum_{t=1}^{L_0}\delta_{s_t}\right),
    \label{eq:bq-scale}
\end{equation}
with a positive lower bound applied in the implementation. The resulting $B_q>0$ is held fixed throughout online evaluation. Thus, expert comparisons account for both interval width and realized miscoverage, so a covered expert may still have a larger total loss if its interval is sufficiently wide.

An expert's coverage feedback controls its own threshold path, but the reported interval uses an average of thresholds. Its miscoverage indicator generally differs from the weighted average of expert miscoverage indicators. To feed this final outcome into expert evaluation, the coverage multiplier $\lambda_t$ is updated using aggregate miscoverage $err_t$:
\begin{equation}
    \lambda_{t}
    =\left[\lambda_{t-1}+\nu(err_{t}-\alpha)\right]_+,
    \qquad \lambda_0=0,
    \label{eq:lambda-update-main}
\end{equation}
where the multiplier step size $\nu>0$ controls the response to aggregate coverage feedback. A miss increases the multiplier, whereas a covered response decreases it toward zero. This adjustment strengthens the miscoverage penalty when coverage falls short of the target and relaxes it when coverage exceeds the target. In the following experiments, we set $\nu=0.02$. The updated $\lambda_t$ enters the current expert losses, which in turn determine the weights for the next round.

Given the losses in \Cref{eq:binary-meta-loss}, we first apply exponential reweighting:
\begin{equation}
    \widetilde w_{t+1}^{(m)}
    =
    \frac{w_t^{(m)}\exp\{-\xi\loss_t^{(m)}\}}
         {\sum_{j=1}^M w_t^{(j)}\exp\{-\xi\loss_t^{(j)}\}}.
    \label{eq:exponential-weight-update}
\end{equation}
Here, $\widetilde w_{t+1}^{(m)}$ is the intermediate normalized weight, and the learning rate $\xi>0$ controls sensitivity to loss differences. The exponential update increases the relative weight of experts with lower loss while preserving the information contained in the current weight distribution.

A pure exponential update may give an expert so little weight that it recovers slowly when its subsequent losses improve. To keep early losses from effectively removing a configuration from future comparisons, we apply a fixed-share step \citep{herbster1998tracking}:
\begin{equation}
    w_{t+1}^{(m)}
    =(1-\sigma)\widetilde w_{t+1}^{(m)}+\frac{\sigma}{M},
    \label{eq:fixed-share-weight-update}
\end{equation}
where $0\le\sigma\le1$ is the share parameter. This update retains a fraction $1-\sigma$ of the allocation determined by the losses while redistributing the remaining mass uniformly across the expert library. For $\sigma>0$, every expert keeps weight at least $\sigma/M$ and remains available for later comparisons.

The resulting weights are then used to aggregate the experts' candidate thresholds at the next time step,
\begin{equation*}
    q_{t+1}=\sum_{m=1}^M w_{t+1}^{(m)}q_{t+1}^{(m)},
\end{equation*}
yielding the prediction interval $[\widehat y_{t+1|t}-q_{t+1},\widehat y_{t+1|t}+q_{t+1}]$.

\Cref{alg:pid-expert} summarizes the procedure. Conditional on the point forecast and shared score information, the computational and storage costs at each round are both $O(M)$, assuming that the controller states have fixed dimension. The loss and PID updates for individual experts can be performed in parallel.

\begin{algorithm}[h]
\caption{PID-Expert aggregation}
\label{alg:pid-expert}
\small
\begin{algorithmic}[1]
\Require Target miscoverage level $\alpha$; PID library $\{\theta_m\}_{m=1}^M$; normalization scale $B_q>0$; learning rate $\xi>0$;
share parameter $0\leq\sigma\leq 1$; multiplier step size $\nu>0$.
\State Initialize all PID states, $w_1^{(m)}\gets1/M$, for $m=1,\ldots,M$, and $\lambda_0\gets0$.
\For{$t=1,\ldots,T$}
    \State Obtain the shared forecast $\widehat y_{t|t-1}$ and candidate thresholds $\{q_t^{(m)}\}_{m=1}^M$.
    \State Set $q_t\gets\sum_m w_t^{(m)}q_t^{(m)}$ and output $[\widehat y_{t|t-1}-q_t,\widehat y_{t|t-1}+q_t]$.
    \State Observe $y_t$ and compute $err_t$ and $\{err_t^{(m)}\}_{m=1}^M$.
    \State Update $\lambda_{t} =\left[\lambda_{t-1}+\nu(err_{t} -\alpha)\right]_+$.
    \State Compute the expert loss $\loss_t^{(m)}=2q_t^{(m)}/B_q+\lambda_{t}err_t^{(m)},\quad m=1,\ldots,M$.
    \State Update the intermediate weights $\widetilde w_{t+1}^{(m)} = \frac{w_t^{(m)}\exp\{-\xi\loss_t^{(m)}\}}{\sum_{j=1}^M w_t^{(j)}\exp\{-\xi\loss_t^{(j)}\}},\quad m=1,\ldots,M$.
    \State Apply the fixed-share update $w_{t+1}^{(m)} =(1-\sigma)\widetilde w_{t+1}^{(m)}+\frac{\sigma}{M},\quad m=1,\ldots,M$.
    \State Update each expert's PID states using its own coverage feedback $err_t^{(m)}$.
\EndFor
\end{algorithmic}
\end{algorithm}

\section{Aggregation properties and coverage analysis}
\label{sec:parameter}

The aggregation rule raises two questions: how the expert weights compare with a fixed configuration, and how the reported interval controls miscoverage. These questions require separate arguments because the weighted average of expert miscoverage indicators generally differs from the miscoverage indicator of the aggregate threshold. We first study the comparison made by the expert weighting layer, then analyze the aggregate coverage feedback and state what the two results imply for interval efficiency.

\subsection{Comparison properties of the expert weights}

Fix an interval $\mathcal I=[r,s]$ of length $L=s-r+1$. For the expert losses in \Cref{eq:binary-meta-loss}, define the weighted first and second moments
\begin{equation*}
    \bar\loss_t
    :=\sum_{m=1}^M w_t^{(m)}\loss_t^{(m)},
    \qquad
    \bar v_t
    :=\sum_{m=1}^M w_t^{(m)}\bigl(\loss_t^{(m)}\bigr)^2.
\end{equation*}
The local regret relative to a fixed expert $m$ is
\begin{equation*}
    \operatorname{Reg}_\mathcal{I}(m)
    :=\sum_{t=r}^s\bar\loss_t
      -\sum_{t=r}^s\loss_t^{(m)}.
\end{equation*}

The comparison uses the same realized sequence of coverage multipliers $\{\lambda_t\}$ for every expert. Thus, it measures how the weighting layer performs relative to a fixed expert under the losses generated by the actual aggregation procedure.

\begin{assumption}[Well-defined expert feedback]
\label{ass:finite-meta-loss}
The expert library is finite, $M<\infty$, and the losses satisfy $0\le\loss_t^{(m)}<\infty$ for all $t$ and $m$. The aggregation parameters are fixed over time with $\xi>0$ and $0<\sigma\le1/2$. The initial weights are uniform, $w_1^{(m)}=1/M$, and subsequent weights are updated according to \Cref{eq:exponential-weight-update,eq:fixed-share-weight-update}.
\end{assumption}

The restriction on $\sigma$ is used for the regret analysis, while the weight update itself is defined for $0\le\sigma\le1$. Positive sharing ensures the uniform lower bound $w_t^{(m)}\ge \frac{\sigma}{M}$, while $\sigma\le1/2$ allows the simplification $-\log(1-\sigma)\le2\sigma$.

\begin{theorem}[Local fixed-expert regret]
\label{thm:local-fixed-expert-regret}
Under \Cref{ass:finite-meta-loss}, for any interval $\mathcal I=[r,s]$ of length $L=s-r+1$ and any expert $m$,
\begin{align}
    \operatorname{Reg}_\mathcal I(m)
    &\le
    \xi\sum_{t=r}^s\bar v_t
    +\frac{\log(M/\sigma)-L\log(1-\sigma)}{\xi}
    \label{eq:local-regret-exact}\\
    &\le
    \xi\sum_{t=r}^s\bar v_t
    +\frac{\log(M/\sigma)+2\sigma L}{\xi}.
    \label{eq:local-regret-simplified}
\end{align}
\end{theorem}

The result follows from a fixed-share comparison argument \citep{herbster1998tracking} and holds pathwise, even when the expert losses depend on aggregate coverage feedback through $\lambda_t$. It therefore requires no assumptions of independence, stationarity, or exchangeability. The bound depends on the accumulated second moments of the expert losses, so finiteness of the individual losses alone is not sufficient to guarantee sublinear regret.

The proof is given in \Cref{app:regret-parameter-proofs}. An extension to an external benchmark sequence, conditional on a library approximation assumption, is stated and proved in \Cref{app:target-sequence-comparison}.

\paragraph{Application to aggregation parameters}

We use the simplified fixed-expert bound in \Cref{eq:local-regret-simplified} to derive practical reference scales for the aggregation parameters. We begin by fixing a local time scale $L\ge1$ and optimizing the $\sigma$-dependent complexity term.

\begin{corollary}[Fixed-share scale]
\label{cor:share-rule}
For a prescribed local interval length $L$, the value minimizing the $\sigma$-dependent complexity term in \Cref{eq:local-regret-simplified} is
\begin{equation*}
    \sigma^\star=\frac{1}{2L}.
\end{equation*}
\end{corollary}

Next, define
\begin{equation*}
    S_\mathcal I:=\sum_{t=r}^s\bar v_t,
    \qquad
    V_\mathcal I:=\frac{S_\mathcal I}{L}.
\end{equation*}
Substituting $\sigma=1/(2L)$ into \Cref{eq:local-regret-simplified} leaves the $\xi$-dependent terms $\xi S_{\mathcal I} + \frac{\log(2ML)+1}{\xi}$. For a reference calculation that holds $S_{\mathcal I}>0$ fixed, minimizing this expression over $\xi>0$ gives
\begin{equation}
    \xi_\mathcal I^\star
    =\sqrt{\frac{\log(2ML)+1}{S_\mathcal I}}
    =\sqrt{\frac{\log(2ML)+1}{LV_\mathcal I}}.
    \label{eq:xi-realized-optimum}
\end{equation}

Since $V_\mathcal I$ is not available when the parameters are chosen, we use a reference scale $V>0$ and set
\begin{equation}
    \boxed{
    \sigma=\frac{1}{2L},
    \qquad
    \xi=\sqrt{\frac{\log(2ML)+1}{LV}}
    }
    \label{eq:aggregation-rules}
\end{equation}
throughout each run.

The choice of $L$ sets the local time scale for these reference parameters; the bound itself holds over arbitrary intervals without reoptimizing them. The scale $V$ is fixed before evaluation using the reference calculation in \Cref{app:reference-v-derivation}. The proof of \Cref{cor:share-rule} and the learning-rate calculation are included in \Cref{app:regret-parameter-proofs}.

\subsection{Coverage feedback for the reported interval}
\label{sec:coverage}

The multiplier recursion in \Cref{eq:lambda-update-main} yields the following pathwise relation between average aggregate miscoverage and multiplier growth.

\begin{proposition}[Pathwise upper bound on aggregate miscoverage]
\label{prop:lambda-certificate}
For every $T\ge1$,
\begin{equation}
    \frac{1}{T}\sum_{t=1}^T err_t
    \le
    \alpha+
    \frac{\lambda_T-\lambda_0}{\nu T}.
    \label{eq:lambda-bound}
\end{equation}
Consequently, along any realized path for which $\lambda_T=o(T)$,
\begin{equation}
    \limsup_{T\to\infty}
    \frac{1}{T}\sum_{t=1}^T err_t
    \le\alpha.
    \label{eq:lambda-asymptotic-upper}
\end{equation}
\end{proposition}

The relation gives asymptotic control when the multiplier grows sublinearly. To obtain a sufficient condition for this behavior, we next impose negative drift once the multiplier is sufficiently large.

Let $\mathcal F_{t-1}$ denote the information available before the response at time $t$ is observed, so that $q_t$ and $\lambda_{t-1}$ are $\mathcal F_{t-1}$-measurable.

\begin{assumption}[High-multiplier negative drift]
\label{ass:high-multiplier-drift}
There exist constants $\nu\alpha\le\Lambda<\infty$ and $0<\delta\le\alpha$ such that, almost surely for every $t\ge1$,
\begin{equation}
    \E\!\left[err_t-\alpha\mid\mathcal F_{t-1}\right]
    \le -\delta
    \qquad\text{whenever }\lambda_{t-1}\ge\Lambda.
    \label{eq:high-multiplier-drift}
\end{equation}
Moreover, the multiplier evolves according to the recursion in \Cref{eq:lambda-update-main} from a fixed finite initial value $\lambda_0\ge0$. The algorithm uses $\lambda_0=0$.
\end{assumption}

This assumption requires the conditional aggregate miscoverage probability to be at most $\alpha-\delta$ whenever $\lambda_{t-1}\ge\Lambda$. Since $\Lambda\ge\nu\alpha$, the projection is inactive in this region, which implies $\E[\lambda_t-\lambda_{t-1}\mid\mathcal F_{t-1}]\le-\nu\delta$. The restriction on $\Lambda$ entails no loss of generality, since any finite threshold above which the condition holds can be increased to at least $\nu\alpha$. A larger multiplier places greater emphasis on miscoverage in the expert losses and thereby influences subsequent weights. This negative-drift property is an additional assumption on the combined expert library, controller dynamics, and data-generating process.

\begin{proposition}[Mean-square sublinearity of the multiplier]
\label{prop:multiplier-mean-square}
Under \Cref{ass:high-multiplier-drift}, define
\begin{equation*}
    C_\lambda
    :=\nu^2+2\nu\Lambda(1-\alpha).
\end{equation*}
Then, for every $T\ge1$,
\begin{equation}
    \E[\lambda_T^2]
    \le \lambda_0^2+C_\lambda T.
    \label{eq:lambda-second-moment-bound}
\end{equation}
Consequently,
\begin{equation}
    \E\!\left[\left(\frac{\lambda_T}{T}\right)^2\right]
    \le \frac{\lambda_0^2}{T^2}+\frac{C_\lambda}{T}
    \longrightarrow0.
    \label{eq:lambda-l2-sublinear}
\end{equation}
\end{proposition}

Thus, $\lambda_T/T\to0$ in $L^2$ and hence in probability. Together with the bounded increments of the multiplier recursion, the second-moment bound also yields $\lambda_T/T\to0$ almost surely; see \Cref{app:multiplier-stability-proof}.

Combining \Cref{prop:lambda-certificate,prop:multiplier-mean-square} yields a nonasymptotic bound for expected average miscoverage.

\begin{corollary}[Expected aggregate miscoverage]
\label{cor:expected-miscoverage}
Under \Cref{ass:high-multiplier-drift},
\begin{align}
    \E\!\left[\frac{1}{T}\sum_{t=1}^T err_t\right]
    &\le
    \alpha+\frac{\sqrt{\E[\lambda_T^2]}}{\nu T}
    \notag\\
    &\le
    \alpha+
    \frac{\sqrt{\lambda_0^2+C_\lambda T}}{\nu T}
    =\alpha+O(T^{-1/2}).
    \label{eq:expected-miscoverage-rate}
\end{align}
\end{corollary}

Under the same assumption, the almost-sure sublinearity of the multiplier further implies that \Cref{eq:lambda-asymptotic-upper} holds almost surely. Proofs of the coverage results appear in \Cref{app:multiplier-proof,app:multiplier-stability-proof,app:expected-miscoverage-proof}.

\subsection{Role of the aggregation components}

The pathwise coverage result is obtained from the realized aggregate error after the reported threshold has been formed. The corresponding feedback pathway can be summarized as
\begin{equation*}
    (\{\theta_m\},B_q,\xi,\sigma,\nu)
    \longrightarrow
    \{q_t^{(m)},w_t^{(m)}\}_{m=1}^M
    \longrightarrow q_t
    \longrightarrow err_t
    \longrightarrow\lambda_t
    \longrightarrow\{\loss_t^{(m)},w_{t+1}^{(m)}\}_{m=1}^M.
\end{equation*}
The quantities $\{\theta_m\}$, $B_q$, $\xi$, and $\sigma$ do not appear explicitly in \Cref{eq:lambda-bound} because the pathwise bound follows directly from the scalar multiplier update. They nevertheless remain behaviorally important because, through the expert thresholds and weights, they affect $q_t$, the realized miscoverage sequence, and the evolution of $\lambda_t$. These components therefore influence whether \Cref{ass:high-multiplier-drift} holds and, when it does, the admissible values of the stability constants $\Lambda$ and $\delta$.

More generally, \Cref{prop:lambda-certificate} remains algebraically valid for any predictable rule that generates $q_t$ and subsequently applies the same multiplier update. The expected and almost-sure asymptotic guarantees additionally rely on the high-multiplier negative-drift condition in \Cref{ass:high-multiplier-drift}.

\section{Experiments}
\label{sec:experiments}

In this section, we evaluate the proposed PID-Expert aggregation method through simulations and two real-data applications. The following experiment choices are used throughout:

\begin{enumerate}[label=(\alph*),leftmargin=2.2em,itemsep=0.2em]

    \item All experiments use one-step-ahead forecasting with symmetric
    prediction intervals at the nominal coverage level $1-\alpha=0.9$.

    \item The benchmark methods are configured following \citet{angelopoulos2023pid} for Conformal PID and \citet{gibbs2024} for DtACI and AgACI. AgACI is implemented here by aggregating calibration levels; the original method aggregates interval endpoints separately \citep{zaffran2022}. Further implementation details are provided in \Cref{app:expert-pools,app:aci-benchmarks}. We also include two natural variants of PID-Expert as benchmarks, using the same expert threshold trajectories: PID-Expert selection adaptively uses the threshold from whichever expert receives the highest weight in PID-Expert aggregation at each time step; PID-Expert average takes an arithmetic mean of all expert thresholds. These comparisons change how the final threshold is formed; selection reuses the aggregation weights without a separate weight update.
    
    \item For PID-Expert aggregation, we set the multiplier step size to $\nu=0.02$ and fix the width normalization scale using the initial score window, as defined in \Cref{eq:bq-scale}. We use the prespecified reference value $V=1.3533$ in the parameter rule $\sigma=1/(2L)$ and $\xi=\sqrt{\left(\log(2ML)+1\right)/(LV)}$ from \Cref{eq:aggregation-rules}. The expert library is constructed from a parameter grid around the recommended Conformal PID configuration; the grid is specified in \Cref{app:expert-pools}.
    
    \item We assess coverage using empirical and rolling coverage and interval efficiency using mean and median width. We also report the mean scaled interval score \citep[MSIS,][]{gneiting2007-st}, which accounts jointly for interval width and the magnitude of coverage failures. For a symmetric interval with half-width $q_t$ and absolute forecast error $s_t$, the interval score \citep{winkler1972} is
\begin{equation*}
    \operatorname{IS}_t
    =2q_t+\frac{2}{\alpha}(s_t-q_t)_+.
\end{equation*}
MSIS is the average interval score over the evaluation period divided by the mean absolute error of a seasonal naive forecast for seasonal series, or a one-step naive forecast for nonseasonal series. Lower MSIS values indicate better interval forecasts under this scoring criterion.
    
\end{enumerate}

\subsection{Simulations}

In this subsection, we evaluate the proposed aggregation method using two simulation designs: a nonlinear autoregressive process with exogenous predictors and Gaussian sequences with stationary, gradually changing, or abruptly changing means.

\subsubsection{Nonlinear simulation}

The first experiment follows the nonlinear data-generating process of \citet{wang2024cpts},
\begin{equation*}
    y_t
    =\sin(y_{t-1})
    +0.5\log(y_{t-2}+1)
    +0.1y_{t-1}x_{1,t}
    +0.3x_{2,t}
    +\varepsilon_t,
\end{equation*}
where $x_{1,t},x_{2,t}\stackrel{\mathrm{iid}}{\sim}\operatorname{Unif}(0,1)$ and $\varepsilon_t$ is Gaussian white noise with standard deviation $0.1$. After an appropriate burn-in period, we generate 2,000 data points.

One-step-ahead point forecasts are produced by a feed-forward neural network with a single hidden layer, using lagged responses and the two exogenous predictors as inputs. Following the neural autoregression design of \citet{wang2024cpts}, the model is refitted at each forecast origin using a rolling training window of 500 observations, with a calibration window of the same length. Conformal PID and PID-Expert use the same Theta model scorecaster \citep{assimakopoulos2000}. The PID-Expert library contains $M=343$ controllers. We set the calibration window size $L_0=500$ used to give $B_q$ and $L=100$ as a local time scale, yielding $\sigma=0.005$ and $\xi=0.299$.

\begin{table}[!htbp]
    \centering
    \begin{threeparttable}
    \caption{Performance in the nonlinear simulation in terms of empirical coverage, mean and median interval width, minimum rolling coverage, and MSIS.}
    \label{tab:nonlinear-four-method}
    \small
    \begin{tabular*}{\textwidth}{@{\extracolsep{\fill}}lrrrrr}
\toprule
Method & Coverage & Mean width & Median width & Min roll. cov. & MSIS \\
\midrule
DtACI & 0.892 & \textbf{0.337} & \textbf{0.333} & 0.820 & \textbf{2.453} \\
AgACI & 0.894 & \textbf{0.337} & 0.335 & 0.810 & \textbf{2.453} \\
Conformal PID & 0.904 & 0.360 & 0.357 & \textbf{0.860} & 2.513 \\
PID-Expert Selection & 0.895 & 0.350 & 0.349 & 0.840 & 2.503 \\
PID-Expert Average & 0.909 & 0.363 & 0.358 & \textbf{0.860} & 2.504 \\
PID-Expert Aggregation & \textbf{0.903} & 0.352 & 0.350 & 0.850 & 2.489 \\
\bottomrule
\end{tabular*}

    \begin{tablenotes}[flushleft]\footnotesize
        \item Minimum rolling coverage is computed over a window of 100 observations.
    \end{tablenotes}
    \end{threeparttable}
\end{table}

\Cref{tab:nonlinear-four-method} reports the results for the simulated series. Bold entries indicate empirical coverage closest to the nominal level, the highest minimum rolling coverage, and the smallest width and MSIS values. DtACI and AgACI yield the narrowest intervals and lowest MSIS, but their empirical coverage falls below the nominal level and their minimum rolling coverage is also lower. PID-Expert aggregation achieves empirical coverage of $0.903$, the closest to the nominal level among the six methods. Compared with Conformal PID, it reduces the mean interval width from $0.360$ to $0.352$ and MSIS from $2.513$ to $2.489$, while maintaining a comparable minimum rolling coverage of $0.850$, versus $0.860$ for Conformal PID. PID-Expert selection further reduces the mean width slightly to $0.350$, but its coverage falls to $0.895$. By contrast, PID-Expert average attains higher coverage of $0.909$ but produces a wider interval with a mean width of $0.363$. Both alternatives have higher MSIS than PID-Expert aggregation. Overall, PID-Expert aggregation provides a favorable balance between coverage and interval efficiency.

\Cref{fig:nonlinear-rolling} shows how coverage and interval width evolve over rolling windows of 100 observations. Conformal PID and PID-Expert aggregation exhibit shallower declines in rolling coverage than DtACI and AgACI, with minimum rolling coverage of $0.860$ and $0.850$, compared with $0.820$ and $0.810$, respectively. PID-Expert selection reaches a lower minimum rolling coverage of $0.840$, whereas PID-Expert average matches Conformal PID at $0.860$. In terms of interval width, PID-Expert aggregation generally produces narrower rolling mean and median widths than Conformal PID and PID-Expert average while maintaining rolling coverage close to the target level.

\begin{figure}[!htbp]
    \centering
    \includegraphics[width=0.9\textwidth]{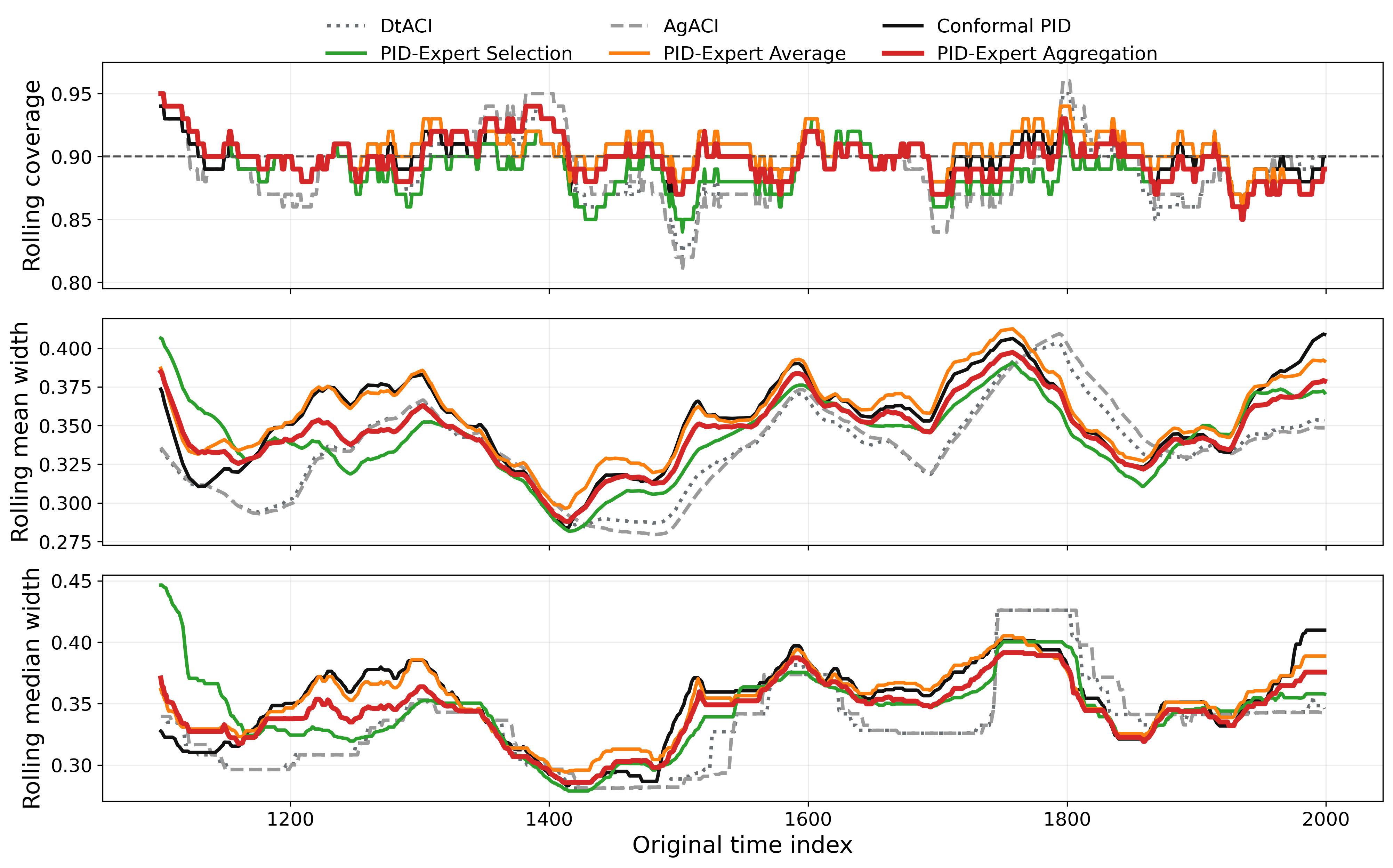}
    \caption{Rolling coverage, mean interval width, and median interval width in the nonlinear simulation, computed over a window of 100 observations. The horizontal dashed line marks the target coverage level of $0.9$.}
    \label{fig:nonlinear-rolling}
\end{figure}

\Cref{fig:nonlinear-expert-scatter} compares PID-Expert aggregation with all 343 candidate PID controllers in the library. Across these experts, empirical coverage ranges from $0.901$ to $0.909$, while mean interval width varies considerably. For reference, we define the hindsight best feasible expert as the controller with the smallest mean width among those achieving empirical coverage of at least $0.9$ on this realization. This expert has a mean width of $0.348$, compared with $0.352$ for PID-Expert aggregation and $0.360$ for Conformal PID. Based on these reported values, PID-Expert aggregation closes approximately $67\%$ of the gap in mean interval width between Conformal PID and the hindsight benchmark. PID-Expert selection produces a slightly narrower interval, with a mean width of $0.350$, but its empirical coverage falls to $0.895$, below the target level. In contrast, PID-Expert average raises coverage to $0.909$ but increases the mean width to $0.363$, exceeding that of Conformal PID. This comparison illustrates how online aggregation can approach the performance of a favorable library member without selecting that controller ex post using the evaluation outcomes.

\begin{figure}[!htbp]
    \centering
    \includegraphics[width=0.9\textwidth]{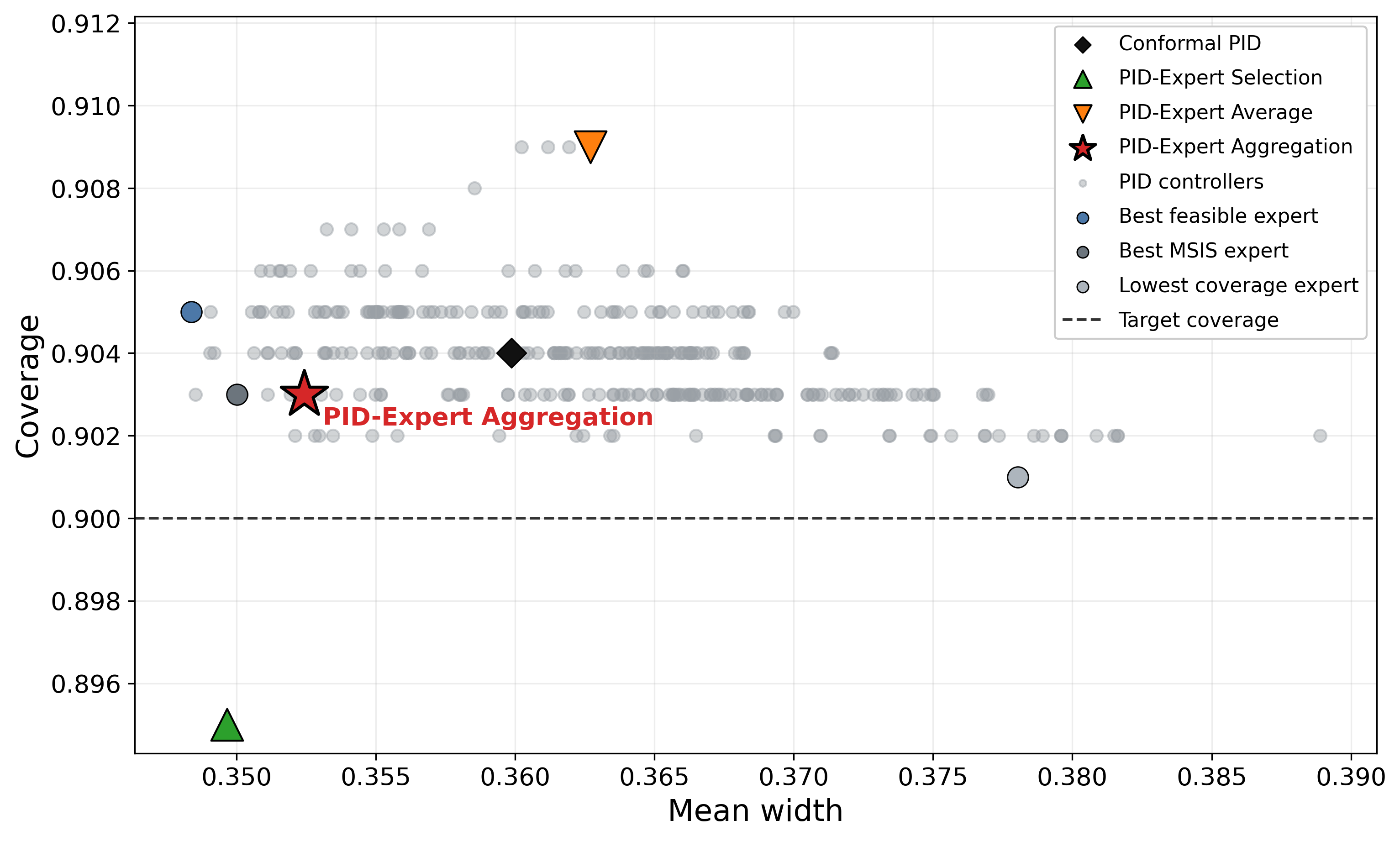}
    \caption{Empirical coverage and mean interval width for the 343 PID controllers, Conformal PID, PID-Expert selection, PID-Expert average, and PID-Expert aggregation in the nonlinear simulation. The best feasible expert is identified ex post from the evaluation sample and therefore represents a hindsight oracle.}
    \label{fig:nonlinear-expert-scatter}
\end{figure}

\Cref{fig:nonlinear-weights} illustrates the evolution of the aggregation weights for the five experts attaining the highest peak weights. E23 receives greater weight early in the evaluation period, followed by E27 over much of the middle portion. The allocation subsequently shifts toward E13, with E136 and E128 gaining greater influence later in the sequence. The experts favored in the earlier periods use $\eta=0.025$, whereas E136 and E128 use $\eta=0.075$ together with larger integral gains $K_I$. These evolving weight patterns illustrate how PID-Expert reallocates influence across heterogeneous controller configurations as their relative performance changes over time, rather than relying on a single prespecified configuration throughout.

\begin{figure}[!htbp]
    \centering
    \includegraphics[width=0.9\textwidth]{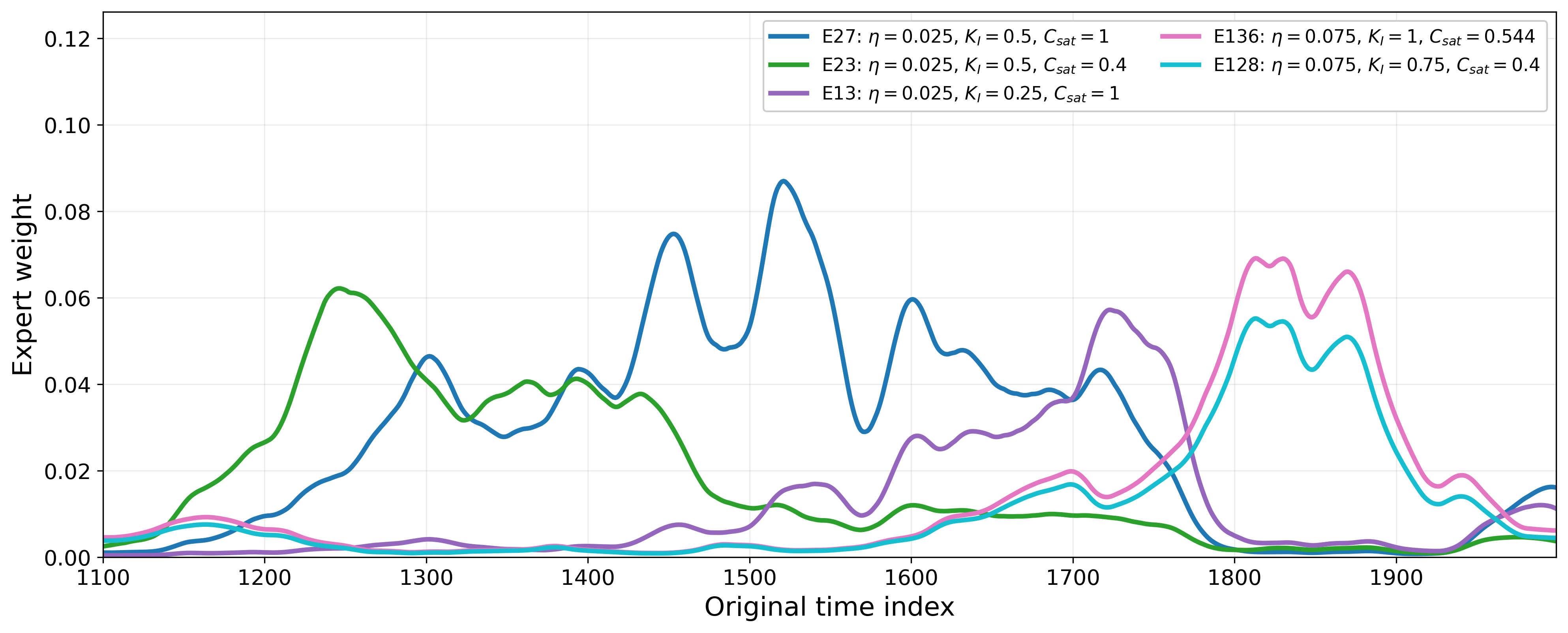}
    \caption{Evolution of the PID-Expert weights for the five experts attaining the highest peak weights in the nonlinear simulation. The displayed trajectories are smoothed using a moving average over 15 time steps for visualization, and labels such as E23 denote expert indices.}
    \label{fig:nonlinear-weights}
\end{figure}

\subsubsection{Distribution-shift simulations}

Following \citet{gibbs2024}, we consider three Gaussian simulation settings with stationary, gradually changing, and abruptly changing means. Conditional on the mean trajectory, observations are independent and satisfy
\begin{equation*}
    Y_t\mid\mu_t\sim N(\mu_t,1).
\end{equation*}
For each setting, we generate a sequence of 6,000 observations. The three settings differ only in the evolution of $\mu_t$.

\textbf{Stationary.}
The mean remains fixed at $\mu_t=0$, so the observations are independent standard normal random variables.

\textbf{Smooth shift.}
The mean evolves gradually through serially correlated increments. At the
update times, it follows
\begin{equation*}
    m_{j+1}
    =m_j+\frac{1}{2}(m_j-m_{j-1})+\frac{1}{2}\varepsilon_j,
    \qquad
    \varepsilon_j\stackrel{\mathrm{iid}}{\sim}N(0,0.006),
\end{equation*}
starting from $m_0=m_1=0$. Each mean value is held fixed for five consecutive observations, so that $\mu_t=m_{\lceil t/5\rceil}$.

\textbf{Jump shift.}
The sequence is divided into three equal segments. Within each segment, the mean alternates every 500 observations between $-0.075$ and $0.075$ in the first and last segments, and between $-1.5$ and $1.5$ in the middle segment. The resulting design introduces two abrupt transitions between regimes with markedly different magnitudes of distributional change.

We adapt these settings to symmetric prediction intervals, using the nonconformity score $s_t=|Y_t|$ and reporting $[-q_t,q_t]$ for all six methods considered. In the jump setting, the distribution of $|Y_t|$ changes with the magnitude of the mean but not with its sign, so sign reversals within a segment leave the score distribution unchanged.

PID-Expert aggregation uses the same library of $M=343$ controllers in all three scenarios, with $L_0=L=500$, giving $\sigma=0.001$ and $\xi=0.143$. Each scenario is evaluated over 100 trials with independent observation noise and a common mean trajectory, as shown in \Cref{fig:gibbs-mu}. Performance is evaluated after an initial period of 500 observations, and rolling metrics use windows of 500 observations.

\begin{figure}[!htbp]
    \centering
    \includegraphics[width=0.92\textwidth]{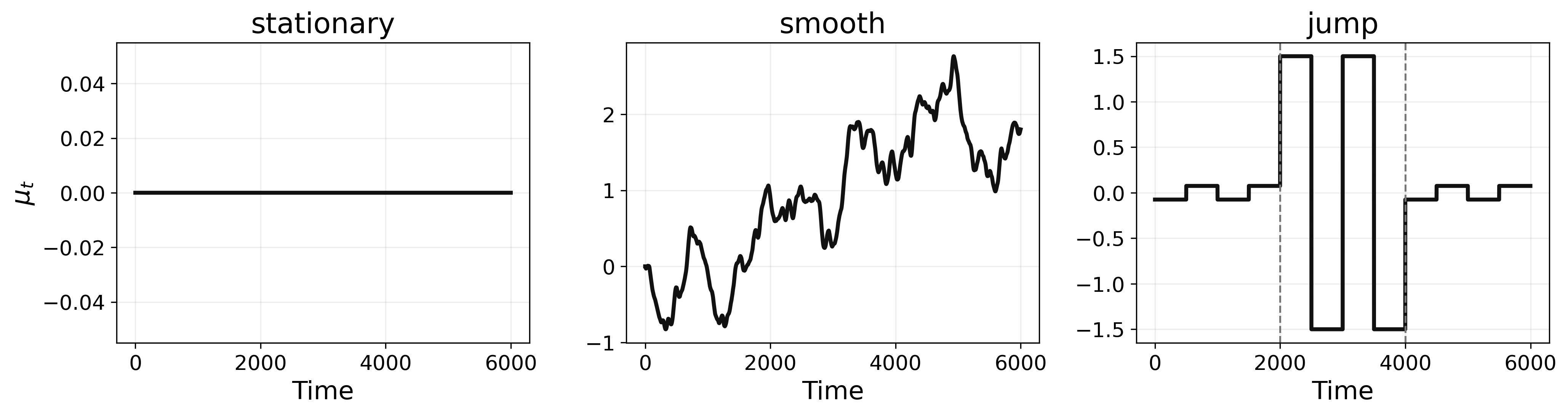}
    \caption{Mean trajectories used in the stationary, smooth-shift, and jump-shift simulation designs.}
    \label{fig:gibbs-mu}
\end{figure}

\begin{table}[!htbp]
    \centering
    \begin{threeparttable}
    \caption{Performance in the stationary, smooth-shift, and jump-shift
    simulations, averaged over 100 trials.}
    \label{tab:gibbs-summary}
    \small
    \begin{tabular*}{\textwidth}{@{\extracolsep{\fill}}llrrrr}
\toprule
Scenario & Method & Coverage & Mean width & Min roll. cov. & MSIS \\
\midrule
Stationary & DtACI & 0.8996 & 3.298 & 0.878 & 3.661 \\
 & AgACI & 0.8997 & \textbf{3.297} & 0.875 & \textbf{3.658} \\
 & Conformal PID & \textbf{0.8998} & 3.414 & \textbf{0.892} & 3.788 \\
 & PID-Expert Selection & 0.8983 & 3.313 & 0.884 & 3.694 \\
 & PID-Expert Average & 0.9034 & 3.420 & 0.891 & 3.757 \\
 & PID-Expert Aggregation & 0.8997 & 3.338 & 0.886 & 3.705 \\
\midrule
Smooth & DtACI & 0.9057 & 5.802 & 0.870 & 5.940 \\
 & AgACI & 0.9032 & 5.687 & 0.870 & 5.862 \\
 & Conformal PID & \textbf{0.8990} & 5.084 & \textbf{0.888} & 5.367 \\
 & PID-Expert Selection & 0.8951 & \textbf{4.945} & 0.867 & \textbf{5.262} \\
 & PID-Expert Average & 0.9025 & 5.091 & 0.887 & 5.323 \\
 & PID-Expert Aggregation & 0.8970 & 4.976 & 0.873 & 5.267 \\
\midrule
Jump & DtACI & 0.9027 & 4.526 & 0.861 & 4.796 \\
 & AgACI & 0.9009 & 4.436 & 0.848 & 4.732 \\
 & Conformal PID & \textbf{0.8997} & 4.259 & \textbf{0.888} & 4.596 \\
 & PID-Expert Selection & 0.8967 & \textbf{4.147} & 0.867 & \textbf{4.504} \\
 & PID-Expert Average & 0.9034 & 4.270 & \textbf{0.888} & 4.558 \\
 & PID-Expert Aggregation & 0.8988 & 4.172 & 0.873 & 4.507 \\
\bottomrule
\end{tabular*}

    \begin{tablenotes}[flushleft]\footnotesize
        \item Minimum rolling coverage over 500-observation windows is calculated within each trial and averaged across 100 trials.
    \end{tablenotes}
    \end{threeparttable}
\end{table}

\Cref{tab:gibbs-summary} summarizes performance averaged over 100 trials. PID-Expert aggregation consistently improves interval efficiency relative to Conformal PID, reducing mean width from $3.414$ to $3.338$ under stationarity, from $5.084$ to $4.976$ under smooth shifts, and from $4.259$ to $4.172$ under jump shifts. MSIS is also lower in all three scenarios, while empirical coverage remains close to the nominal level at $0.8997$, $0.8970$, and $0.8988$, respectively.

Among the PID-Expert variants, selection produces the narrowest intervals and the lowest MSIS under smooth and jump shifts, but also has the lowest empirical and minimum rolling coverage. In contrast, PID-Expert average maintains higher coverage and minimum rolling coverage, but produces slightly wider intervals than Conformal PID in all three scenarios. PID-Expert aggregation lies between these two alternatives, retaining much of the width reduction from selection while providing better coverage. Under stationarity, DtACI and AgACI achieve the smallest widths and MSIS with coverage close to $0.9$. Under smooth and jump shifts, however, both methods produce substantially wider intervals and higher MSIS than the PID-based methods, despite coverage at or above the nominal level.

\begin{figure}[!htbp]
    \centering
    \begin{subfigure}{0.96\textwidth}
        \centering
        \includegraphics[width=\textwidth]{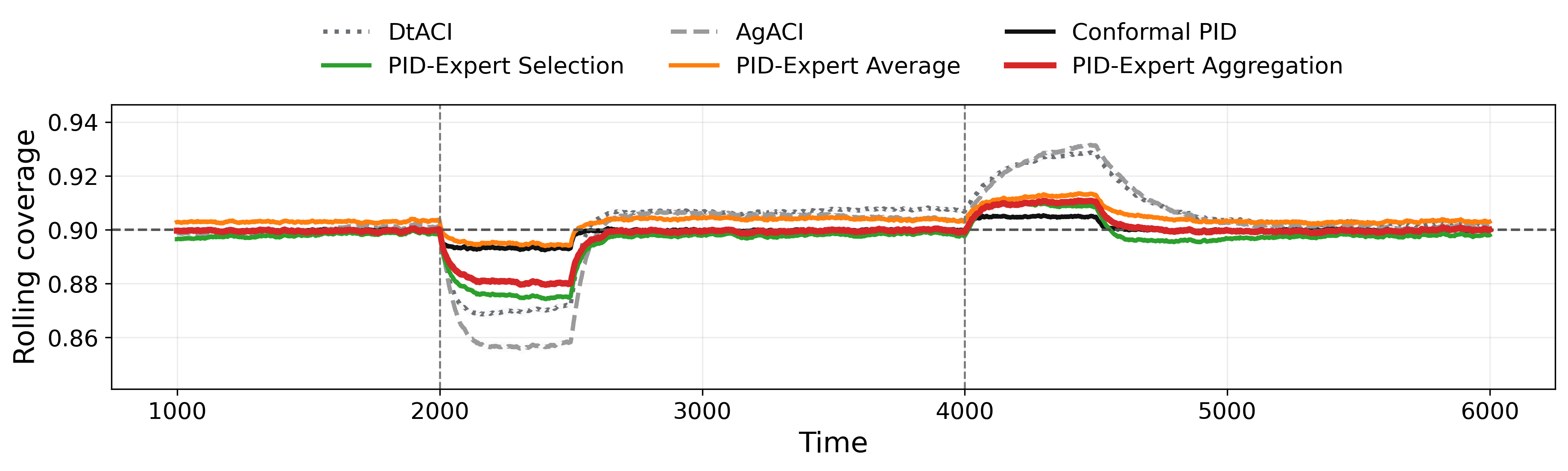}
        \caption{Rolling coverage}
    \end{subfigure}
    \vspace{0.15cm}
    \begin{subfigure}{0.96\textwidth}
        \centering
        \includegraphics[width=\textwidth]{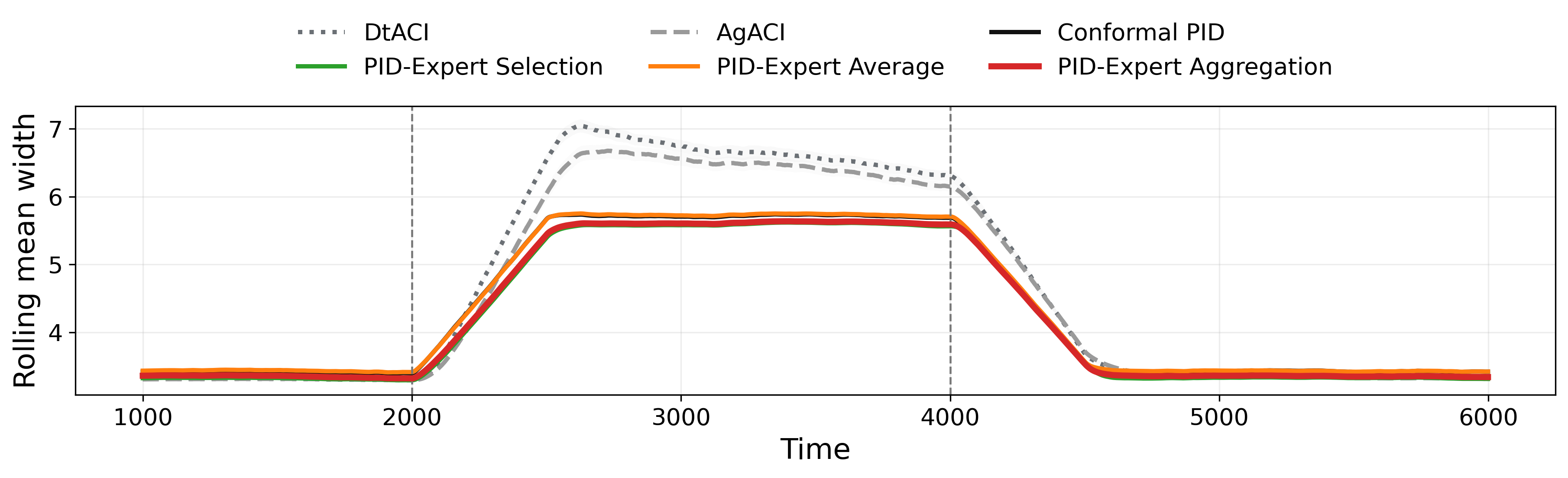}
        \caption{Rolling mean width}
    \end{subfigure}
    \caption{Jump-shift simulation: coverage and mean interval width over rolling windows of 500 observations, averaged over 100 trials. Shaded regions show approximate pointwise 95\% Monte Carlo confidence intervals for the mean trajectories. The horizontal dashed line marks nominal coverage of 0.9, and vertical dashed lines mark changes in mean magnitude.}
    \label{fig:gibbs-jump}
\end{figure}

\Cref{fig:gibbs-jump} shows the rolling coverage and mean width in the jump scenario, averaged over trials. The increase in mean magnitude leads to a temporary decline in coverage for all six methods. This decline is smaller for the PID-based methods than for DtACI and AgACI, although PID-Expert aggregation experiences a somewhat larger dip than Conformal PID. PID-Expert selection exhibits a deeper coverage decline than PID-Expert aggregation and produces slightly narrower intervals, whereas PID-Expert average shows a shallower decline but wider intervals. When the mean magnitude decreases, the ACI variants exhibit a larger and more persistent coverage overshoot. Throughout most of the high-magnitude regime, PID-Expert aggregation closely tracks PID-Expert selection in mean width and generally remains slightly below PID-Expert average and Conformal PID.

\subsection{Real-data applications}

In this subsection, we use two real-world datasets to evaluate the proposed method. Following \citet{montero2020fforma}, we consider six forecasting methods: random walk with drift, automatic exponential smoothing (Auto ETS), feed-forward neural networks with a single hidden layer (NNAR), Theta, TBATS, and STLM-AR. These methods are implemented using the \textit{forecast} package in R \citep{hyndman2024}. To examine performance under different forecasting specifications, we report results for the first three individual models and for an equally weighted average forecast (Average model) constructed from all six models.

\begin{figure}[!htbp]
    \centering
    \begin{subfigure}{0.98\textwidth}
        \centering
        \includegraphics[width=\textwidth]{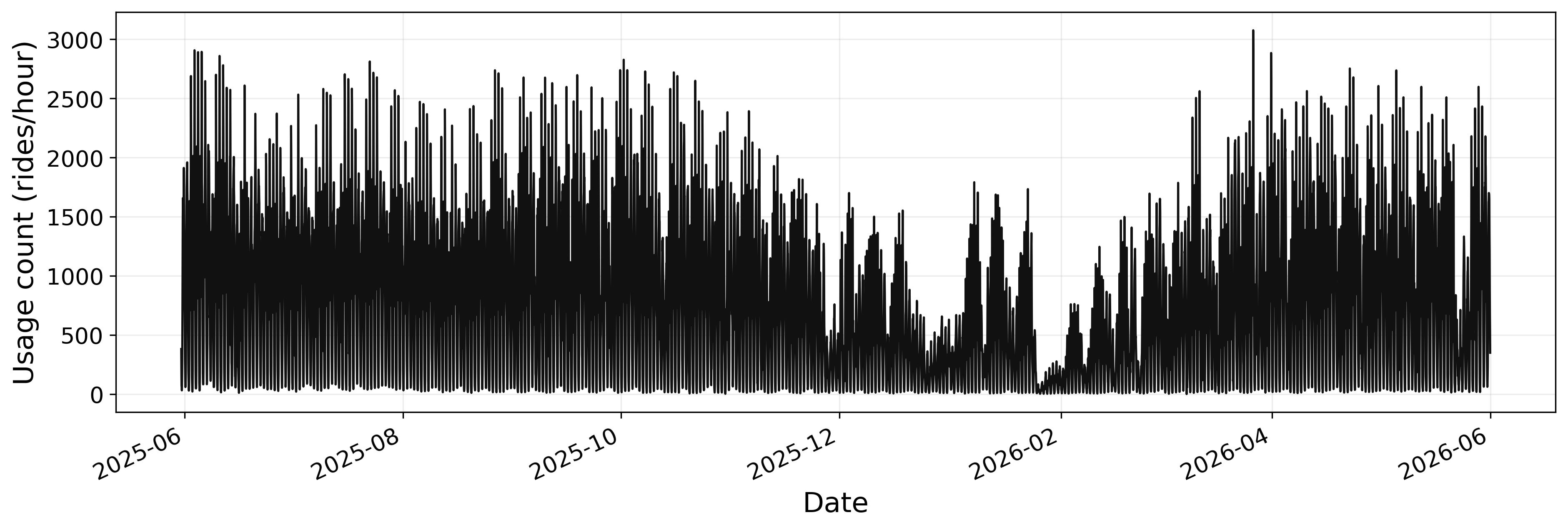}
        \caption{Capital Bikeshare hourly usage from 31 May 2025 to 31 May 2026.}
    \end{subfigure}
    \vspace{0.1cm}
    \begin{subfigure}{0.98\textwidth}
        \centering
        \includegraphics[width=\textwidth]{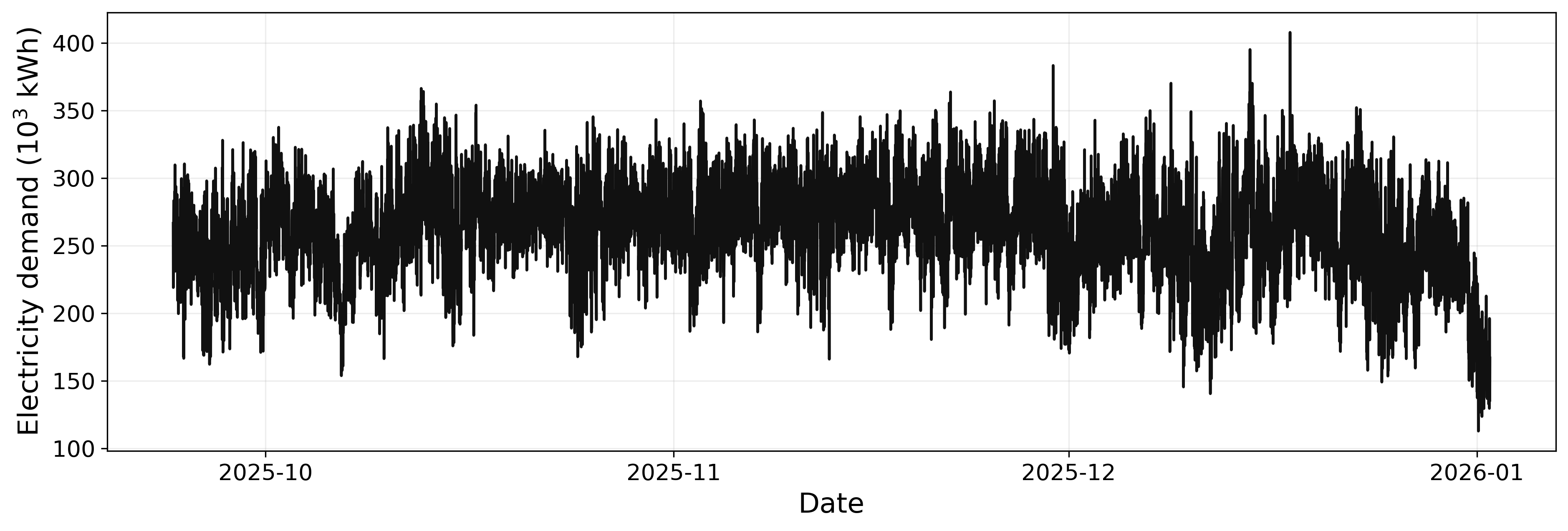}
        \caption{Ordos electricity demand from 24 September 2025 to 1 January 2026.}
    \end{subfigure}
    \caption{Observed Capital Bikeshare usage and Ordos electricity demand.}
    \label{fig:real-series}
\end{figure}

\subsubsection{Capital Bikeshare data}

The Capital Bikeshare dataset contains hourly trip counts for the Washington, DC metropolitan area. We use 8,784 hourly observations spanning from 00:00 on 31 May 2025 to 23:00 on 31 May 2026. As shown in the upper panel of \Cref{fig:real-series}, the series exhibits pronounced high-frequency periodicity together with substantial seasonal variation in both its level and volatility. Usage is generally higher during the warmer months, declines markedly over the winter, and rises again in spring. The pronounced changes in the level and amplitude of the series provide a challenging setting for online prediction interval calibration. We use the first 28 days to train the forecasting models and set $L=L_0=168$, corresponding to one week of hourly observations. This yields $\sigma=0.00298$ and $\xi=0.243$. Interval widths are measured in rides per hour. The data are publicly available at \url{https://capitalbikeshare.com/system-data}.

\begin{table}[!htbp]
    \centering
    \begin{threeparttable}
    \caption{Capital Bikeshare results across the four forecasting specifications, reporting empirical coverage, mean and median interval widths, and MSIS. Widths are measured in rides per hour.}
    \label{tab:bikeshare-results}
    \scriptsize
    \begin{tabular*}{\textwidth}{@{\extracolsep{\fill}}llrrrr}
\toprule
Forecasting model & Method & Coverage & Mean width & Median width & MSIS \\
\midrule
Random walk with drift & DtACI & 0.896 & 881.2 & 906.5 & 6.00 \\
 & AgACI & 0.906 & 899.1 & \textbf{886.7} & \textbf{5.91} \\
 & Conformal PID & \textbf{0.899} & 910.8 & 920.0 & 6.09 \\
 & PID-Expert Selection & 0.888 & \textbf{857.9} & 898.4 & 6.21 \\
 & PID-Expert Average & 0.911 & 907.1 & 916.1 & 5.98 \\
 & PID-Expert Aggregation & 0.898 & 866.8 & 898.3 & 6.03 \\
\midrule
Auto ETS & DtACI & 0.896 & 577.0 & 619.2 & 4.65 \\
 & AgACI & \textbf{0.900} & 590.4 & 620.9 & \textbf{4.59} \\
 & Conformal PID & 0.899 & 603.7 & 614.0 & 4.76 \\
 & PID-Expert Selection & 0.886 & \textbf{561.8} & \textbf{597.8} & 4.72 \\
 & PID-Expert Average & 0.907 & 609.8 & 624.7 & 4.68 \\
 & PID-Expert Aggregation & 0.897 & 570.6 & 602.5 & 4.65 \\
\midrule
NNAR & DtACI & 0.894 & 334.0 & \textbf{313.9} & 2.84 \\
 & AgACI & 0.903 & 349.1 & 320.4 & \textbf{2.81} \\
 & Conformal PID & 0.899 & 355.0 & 327.4 & 2.89 \\
 & PID-Expert Selection & 0.890 & \textbf{326.8} & 318.7 & 2.89 \\
 & PID-Expert Average & 0.912 & 358.3 & 331.9 & 2.83 \\
 & PID-Expert Aggregation & \textbf{0.900} & 332.2 & 322.2 & 2.84 \\
\midrule
Average model & DtACI & 0.896 & 446.2 & 466.9 & \textbf{3.56} \\
 & AgACI & \textbf{0.900} & 461.9 & \textbf{463.4} & 3.58 \\
 & Conformal PID & 0.899 & 477.7 & 484.1 & 3.65 \\
 & PID-Expert Selection & 0.890 & \textbf{440.5} & 471.0 & 3.59 \\
 & PID-Expert Average & 0.906 & 482.9 & 491.2 & 3.61 \\
 & PID-Expert Aggregation & 0.896 & 446.1 & 471.4 & 3.57 \\
\bottomrule
\end{tabular*}

    \end{threeparttable}
\end{table}

As shown in \Cref{tab:bikeshare-results}, PID-Expert aggregation consistently produces narrower mean intervals than DtACI, AgACI, and Conformal PID across all four forecasting specifications, while maintaining empirical coverage close to the nominal level of $0.9$. Relative to Conformal PID, it reduces mean interval width by $4.8\%$, $5.5\%$, $6.4\%$, and $6.6\%$ for the random walk with drift, Auto ETS, NNAR, and Average model, respectively, with only small changes in empirical coverage. It also achieves lower MSIS than Conformal PID in all four cases. PID-Expert selection yields the smallest mean widths, but at the cost of lower coverage, ranging from $0.886$ to $0.890$. In contrast, PID-Expert average attains higher coverage, ranging from $0.906$ to $0.912$, but produces wider intervals. PID-Expert aggregation lies between these two alternatives, keeping coverage closer to the target while retaining much of the width reduction achieved by selection. Although it does not uniformly attain the smallest median width or MSIS, both measures remain competitive in each case.

\begin{figure}[!htbp]
    \centering
    \includegraphics[width=0.98\textwidth]{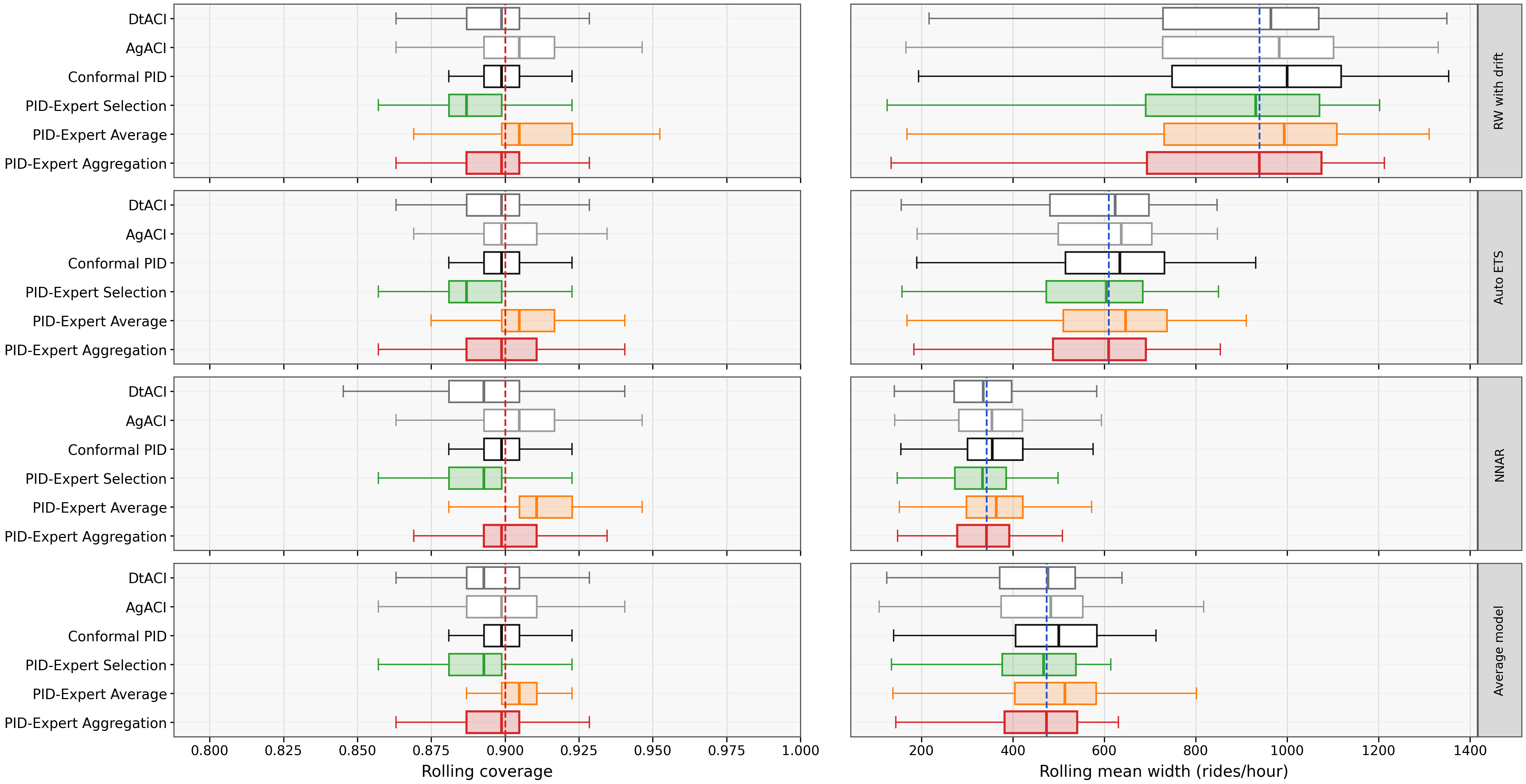}
    \caption{Rolling coverage and mean interval width for the Capital Bikeshare application. Dashed lines indicate the nominal coverage level of 0.9 and the median rolling mean width of PID-Expert aggregation.}
    \label{fig:bikeshare-rolling-boxplots}
\end{figure}

\Cref{fig:bikeshare-rolling-boxplots} further examines the distribution of rolling performance. Across all four forecasting specifications, the rolling mean width distribution of PID-Expert aggregation is consistently shifted below that of Conformal PID, while its rolling coverage remains centered close to the nominal level of 0.9. The three PID-Expert variants also reveal a clear coverage--width trade-off. PID-Expert selection generally produces the narrowest intervals, but its rolling coverage distributions are systematically shifted below the target. PID-Expert average shifts coverage upward, but at the cost of noticeably wider intervals. PID-Expert aggregation lies between these two alternatives, retaining much of the width reduction from selection while bringing coverage closer to the target than selection and producing narrower intervals than averaging. The consistency of this pattern across the four forecasting models suggests that the gain from aggregation is not specific to a particular forecasting specification.

\subsubsection{Electricity demand in Ordos}

Accurate electricity demand forecasting is an important component of power system operation and planning, with implications for power dispatch, generation scheduling, and grid reliability \citep{Weron2014,kath2021conformal,Devine2026,Lipiecki2026}. The electricity demand dataset consists of operational metering records from 153 industrial enterprises above designated size in Ordos, China, sampled at 15-minute intervals. We use 9,600 observations spanning from 00:00 on 24 September 2025 to 23:45 on 1 January 2026. Electricity demand, measured in kWh, is used as the response variable. The Ordos electricity demand data used in this study are available from the corresponding author upon request.

As shown in the lower panel of \Cref{fig:real-series}, the electricity series exhibits substantial high-frequency variation around a relatively stable level over most of the sample period. The variability becomes more pronounced in parts of December, followed by a marked decline in demand toward the end of the observation period. These time-varying features provide a useful setting for evaluating sequential prediction intervals under changing demand conditions. We use the first 14 days to train the forecasting models described above and set the local adaptation and initial score windows to $L=L_0=672$, corresponding to one week of 15-minute observations. This yields $\sigma=0.000744$ and $\xi=0.128$. Interval widths are reported in units of $10^3$ kWh.

\Cref{tab:energy-results} reports the results across the four forecasting specifications. PID-Expert aggregation consistently produces narrower mean intervals than Conformal PID while maintaining empirical coverage close to the nominal level of $0.9$. Relative to Conformal PID, it reduces mean interval width by approximately $1.7\%$--$3.3\%$ and MSIS by $1.2\%$--$2.4\%$. PID-Expert selection yields still narrower mean intervals, but with lower coverage under the random walk with drift, Auto ETS, and NNAR specifications. In contrast, PID-Expert average is more conservative, with coverage of $0.907$--$0.908$ and wider intervals than both Conformal PID and PID-Expert aggregation.

Comparisons with DtACI and AgACI vary across forecasting models. For the random walk with drift and Auto ETS specifications, AgACI attains the lowest MSIS, although PID-Expert aggregation remains competitive in interval width. Under NNAR, PID-Expert aggregation shares the lowest reported MSIS with DtACI and AgACI, while under the Average model it achieves both a narrower mean interval and a lower MSIS than the two ACI methods. Overall, the results show that PID-Expert aggregation improves interval efficiency relative to Conformal PID consistently across forecasting specifications, while maintaining coverage close to the target.

\begin{table}[!htbp]
    \centering
    \begin{threeparttable}
    \caption{Ordos electricity demand results across the four forecasting specifications. The table reports empirical coverage, mean and median interval widths, and the mean scaled interval score (MSIS). Interval widths are measured in $10^3$ kWh.}
    \label{tab:energy-results}
    \scriptsize
    \begin{tabular*}{\textwidth}{@{\extracolsep{\fill}}llrrrr}
\toprule
Forecasting model & Method & Coverage & Mean width & Median width & MSIS \\
\midrule
Random walk with drift & DtACI & 0.899 & \textbf{127.5} & \textbf{126.0} & 4.83 \\
 & AgACI & 0.903 & 129.0 & 126.1 & \textbf{4.79} \\
 & Conformal PID & \textbf{0.900} & 130.4 & 128.5 & 4.88 \\
 & PID-Expert Selection & 0.898 & \textbf{127.5} & 128.3 & 4.89 \\
 & PID-Expert Average & 0.907 & 131.4 & 129.7 & 4.80 \\
 & PID-Expert Aggregation & 0.901 & 128.2 & 128.5 & 4.82 \\
\midrule
Auto ETS & DtACI & 0.901 & 89.0 & \textbf{86.2} & 3.45 \\
 & AgACI & \textbf{0.900} & 89.1 & 86.5 & \textbf{3.44} \\
 & Conformal PID & \textbf{0.900} & 92.0 & 89.8 & 3.53 \\
 & PID-Expert Selection & 0.898 & \textbf{88.6} & 87.2 & 3.45 \\
 & PID-Expert Average & 0.908 & 92.5 & 90.0 & 3.48 \\
 & PID-Expert Aggregation & 0.899 & 89.0 & 87.7 & 3.45 \\
\midrule
NNAR & DtACI & \textbf{0.899} & 83.4 & \textbf{78.6} & \textbf{3.23} \\
 & AgACI & \textbf{0.899} & 83.8 & 79.1 & \textbf{3.23} \\
 & Conformal PID & \textbf{0.899} & 85.5 & 81.9 & 3.29 \\
 & PID-Expert Selection & 0.896 & \textbf{83.2} & 78.9 & 3.25 \\
 & PID-Expert Average & 0.907 & 86.2 & 82.2 & 3.25 \\
 & PID-Expert Aggregation & \textbf{0.899} & 83.6 & 80.0 & \textbf{3.23} \\
\midrule
Average model & DtACI & \textbf{0.900} & 84.4 & \textbf{80.7} & 3.26 \\
 & AgACI & 0.901 & 84.5 & 81.0 & 3.26 \\
 & Conformal PID & \textbf{0.900} & 86.9 & 83.7 & 3.33 \\
 & PID-Expert Selection & \textbf{0.900} & \textbf{83.6} & 81.3 & 3.26 \\
 & PID-Expert Average & 0.908 & 87.5 & 84.8 & 3.29 \\
 & PID-Expert Aggregation & \textbf{0.900} & 84.0 & 82.0 & \textbf{3.25} \\
\bottomrule
\end{tabular*}

    \end{threeparttable}
\end{table}

\Cref{fig:energy-rolling-boxplots} further reveals the coverage--width trade-off at the rolling-window level. Across all four forecasting specifications, PID-Expert aggregation keeps the rolling coverage distribution centered close to the nominal level of $0.9$, while its rolling mean widths are consistently lower than those of Conformal PID. The lower median rolling width supports the efficiency gains reported in \Cref{tab:energy-results}. The three PID-Expert variants display a clear ordering: selection generally attains the narrowest intervals but shifts the rolling coverage distribution downward, whereas averaging shifts coverage upward at the cost of wider intervals. PID-Expert aggregation lies between these two extremes, retaining much of the width reduction from selection and generally raising rolling coverage, while producing narrower intervals than averaging. The same pattern across the four forecasting models highlights the stability of this coverage--efficiency trade-off across different underlying forecasts.

\begin{figure}[!htbp]
    \centering
    \includegraphics[width=0.98\textwidth]{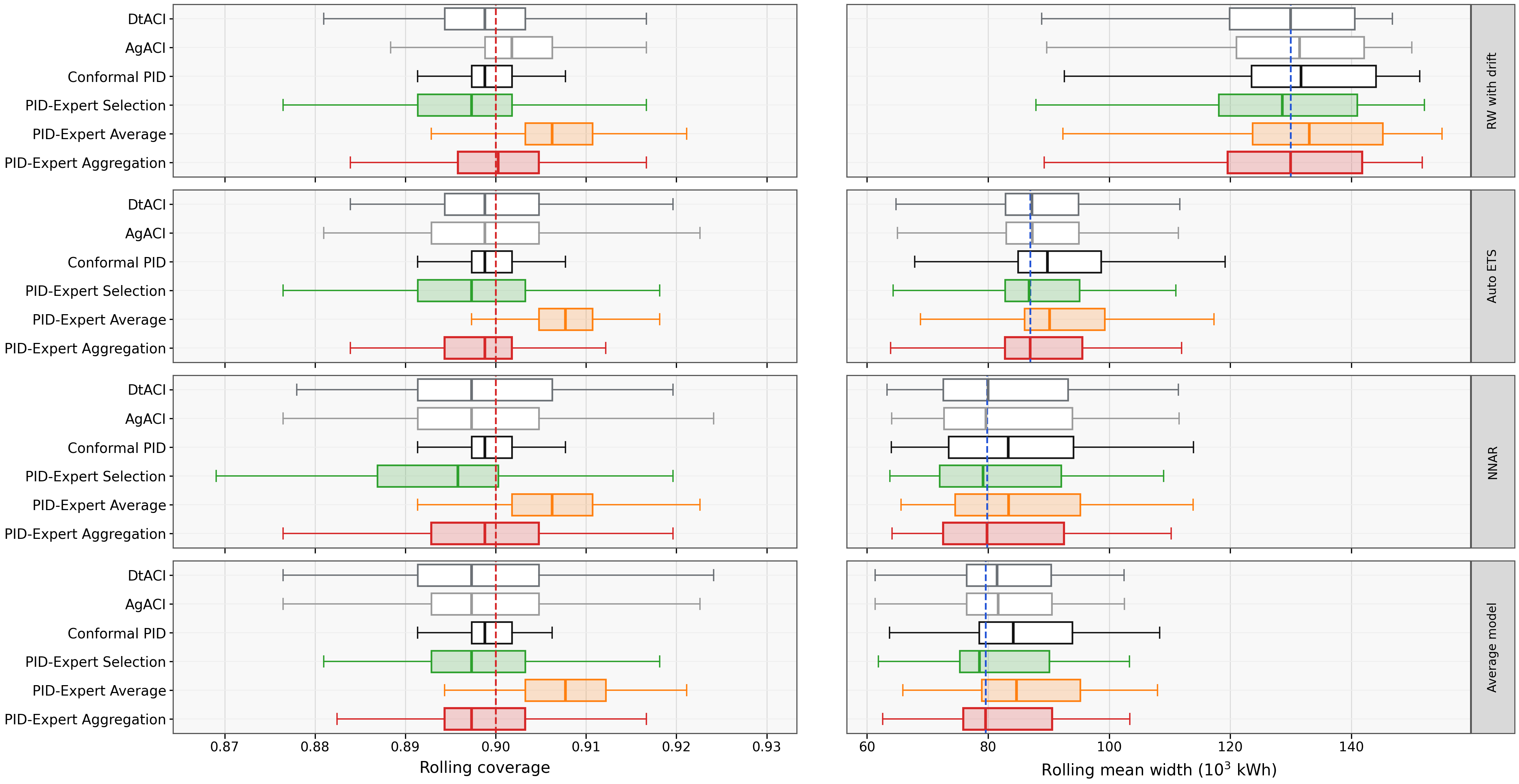}
    \caption{Rolling coverage and mean interval width for the Ordos electricity demand application, computed over windows of 672 observations. The red dashed lines indicate the nominal coverage level of 0.9, and the blue dashed lines indicate the median rolling mean width of PID-Expert aggregation. Interval widths are measured in $10^3$ kWh.}
    \label{fig:energy-rolling-boxplots}
\end{figure}

\FloatBarrier

\section{Conclusions}
\label{sec:conclusion}

In this paper, we propose PID-Expert aggregation, an online conformal framework that adaptively combines the score thresholds generated by a heterogeneous library of PID controllers. Rather than relying on a single prespecified controller, the proposed method updates expert weights through a loss that jointly accounts for normalized interval width and expert-level miscoverage. An aggregate coverage multiplier further adjusts the relative penalty on miscoverage in response to realized coverage errors. This threshold-level aggregation provides a direct mechanism for adapting the coverage--efficiency tradeoff while preserving the distinct calibration dynamics of the individual controllers.

The theoretical analysis characterizes both the aggregation and coverage properties of the method. Local regret bounds compare the weighted average of expert losses with any fixed expert over arbitrary time intervals and motivate practical reference choices for the learning rate and share parameter. Separately, the multiplier recursion yields a deterministic, pathwise upper bound on time-averaged aggregate miscoverage. Under a high-multiplier negative-drift condition, the multiplier grows sublinearly, leading to a nonasymptotic expected miscoverage bound of order $O(T^{-1/2})$ above the target and almost-sure asymptotic upper control.

Across four simulation settings and two real-world forecasting applications, PID-Expert consistently produces narrower mean prediction intervals than the prespecified Conformal PID benchmark while maintaining empirical coverage close to the nominal level. Its relative performance against DtACI and AgACI varies across settings, indicating that aggregating score thresholds and aggregating adaptive miscoverage levels offer complementary advantages across different data regimes. The two additional PID-Expert benchmarks clarify the coverage--efficiency tradeoff. Selection generally produces narrower intervals with lower coverage, whereas equal averaging produces wider intervals with higher coverage. Aggregation retains much of the width reduction from selection while generally bringing empirical coverage closer to the nominal level, but it does not uniformly achieve the lowest MSIS or eliminate local coverage shortfalls. Overall, the results support online aggregation as a way to reduce reliance on a single PID configuration while balancing interval width and coverage.

The present study is limited to one-step-ahead symmetric intervals based on absolute-error scores and a fixed, finite controller library. Its performance therefore depends on whether the library contains controllers with suitable score scales and adaptation dynamics. Future work could extend the framework to asymmetric tail-specific intervals, horizon-dependent aggregation for multistep forecasting, and data-driven procedures for constructing or expanding the controller library online. Such extensions would broaden the applicability of PID-Expert aggregation while retaining its central principle of jointly adapting calibration and interval efficiency.
Proofs and implementation details are given in the appendices. The implementation and available data are provided at \url{https://github.com/Guodong-Liu/PID_Expert}.

\section*{Acknowledgments}

This research was supported by an IIF-SAS Grant, which promotes research in forecasting. Xiaoqian Wang is supported by the National Natural Science Foundation of China (No. 72601399) and the Presidential Foundation of the Academy of Mathematics and Systems Science, Chinese Academy of Sciences, China (No. E555930101). Yanfei Kang is supported by the National Natural Science Foundation of China (No. 72571014).

\section*{Declaration of generative AI use}

During the preparation of this work the authors used ChatGPT to improve English language and assist in coding to create graphical illustrations. After using these tools, the authors reviewed and edited the content as needed and take full responsibility for the content of the publication.

\bibliographystyle{agsm}
\bibliography{references}

@article{angelopoulos2023pid,
  author       = {Angelopoulos, Anastasios N. and Cand\`es, Emmanuel J. and Tibshirani, Ryan J.},
  journal      = {Advances in Neural Information Processing Systems},
  pages        = {23047--23074},
  title        = {Conformal {PID} Control for Time Series Prediction},
  volume       = {36},
  year         = {2023}
}

@article{angelopoulos2023gentle,
  author       = {Angelopoulos, Anastasios N. and Bates, Stephen},
  journal      = {Foundations and Trends in Machine Learning},
  number       = {4},
  pages        = {494--591},
  title        = {Conformal Prediction: A Gentle Introduction},
  volume       = {16},
  year         = {2023}
}

@inproceedings{areces2025optimization,
  author       = {Areces, Felipe and Mohri, Christopher and Hashimoto, Tatsunori and Duchi, John},
  booktitle    = {Proceedings of the 42nd International Conference on Machine Learning},
  pages        = {1604--1649},
  publisher    = {PMLR},
  series       = {Proceedings of Machine Learning Research},
  title        = {Online Conformal Prediction via Online Optimization},
  volume       = {267},
  year         = {2025}
}

@article{assimakopoulos2000,
  author       = {Assimakopoulos, Vassilis and Nikolopoulos, Konstantinos},
  journal      = {International Journal of Forecasting},
  number       = {4},
  pages        = {521--530},
  publisher    = {Elsevier},
  title        = {The theta model: a decomposition approach to forecasting},
  volume       = {16},
  year         = {2000}
}

@article{barber2023,
  author       = {Barber, Rina Foygel and Cand\`es, Emmanuel J. and Ramdas, Aaditya and Tibshirani, Ryan J.},
  journal      = {The Annals of Statistics},
  number       = {2},
  pages        = {816--845},
  publisher    = {Institute of Mathematical Statistics},
  title        = {Conformal Prediction beyond Exchangeability},
  volume       = {51},
  year         = {2023}
}

@article{chatfield1993,
  author       = {Chatfield, Chris},
  journal      = {Journal of Business \& Economic Statistics},
  number       = {2},
  pages        = {121--135},
  publisher    = {Taylor \& Francis},
  title        = {Calculating interval forecasts},
  volume       = {11},
  year         = {1993}
}

@inproceedings{chernozhukov2018,
  author       = {Chernozhukov, Victor and W{\"u}thrich, Kaspar and Yinchu, Zhu},
  booktitle    = {Conference on Learning Theory},
  organization = {PMLR},
  pages        = {732--749},
  title        = {Exact and robust conformal inference methods for predictive machine learning with dependent data},
  year         = {2018}
}

@article{Devine2026,
  author  = {Mel T. Devine and Valentin Bertsch},
  doi     = {10.1016/j.ejor.2026.01.006},
  issn    = {0377-2217},
  journal = {European Journal of Operational Research},
  title   = {Analysing the interactions between demand side and supply side investment decisions in an oligopolistic electricity market using a stochastic equilibrium model},
  year    = {2026},
  note    = {In press}
}

@article{gibbs2021,
  author       = {Gibbs, Isaac and Cand{\`e}s, Emmanuel J},
  journal      = {Advances in Neural Information Processing Systems},
  pages        = {1660--1672},
  title        = {Adaptive conformal inference under distribution shift},
  volume       = {34},
  year         = {2021}
}

@article{gibbs2024,
  author       = {Gibbs, Isaac and Cand{\`e}s, Emmanuel J},
  journal      = {Journal of Machine Learning Research},
  number       = {162},
  pages        = {1--36},
  title        = {Conformal inference for online prediction with arbitrary distribution shifts},
  volume       = {25},
  year         = {2024}
}

@article{gneiting2007,
  author       = {Gneiting, Tilmann and Balabdaoui, Fadoua and Raftery, Adrian E},
  journal      = {Journal of the Royal Statistical Society Series B: Statistical Methodology},
  number       = {2},
  pages        = {243--268},
  publisher    = {Oxford University Press},
  title        = {Probabilistic forecasts, calibration and sharpness},
  volume       = {69},
  year         = {2007}
}

@article{gneiting2007-st,
  author       = {Gneiting, Tilmann and Raftery, Adrian E},
  journal      = {Journal of the American Statistical Association},
  number       = {477},
  pages        = {359--378},
  publisher    = {Taylor \& Francis},
  title        = {Strictly proper scoring rules, prediction, and estimation},
  volume       = {102},
  year         = {2007}
}

@article{gneiting2014,
  author       = {Gneiting, Tilmann and Katzfuss, Matthias},
  journal      = {Annual Review of Statistics and Its Application},
  number       = {1},
  pages        = {125--151},
  publisher    = {Annual Reviews},
  title        = {Probabilistic forecasting},
  volume       = {1},
  year         = {2014}
}

@article{Wang2023,
  author    = {Wang, Xiaoqian and Hyndman, Rob J and Li, Feng and Kang, Yanfei},
  title     = {Forecast combinations: {An} over 50-year review},
  journal   = {International Journal of Forecasting},
  year      = {2023},
  volume    = {39},
  number    = {4},
  pages     = {1518--1547},
  publisher = {Elsevier}
}

@inproceedings{hajihashemi2024ensemble,
  author       = {Hajihashemi, Erfan and Shen, Yanning},
  booktitle    = {Advances in Neural Information Processing Systems},
  doi          = {10.52202/079017-3769},
  title        = {Multi-model Ensemble Conformal Prediction in Dynamic Environments},
  volume       = {37},
  year         = {2024}
}

@article{herbster1998tracking,
  author       = {Herbster, Mark and Warmuth, Manfred K},
  journal      = {Machine Learning},
  pages        = {151--178},
  publisher    = {Springer},
  title        = {Tracking the best expert},
  volume       = {32},
  year         = {1998}
}

@inproceedings{hu2026distribution,
  author       = {Hu, Dongjian and Wu, Junxi and Xia, Shu-Tao and Zou, Changliang},
  booktitle    = {International Conference on Learning Representations},
  pages        = {58230--58255},
  title        = {Distribution-informed online conformal prediction},
  volume       = {2026},
  year         = {2026}
}

@book{hyndman2021,
  author       = {Hyndman, Rob J. and Athanasopoulos, George},
  edition      = {3rd},
  publisher    = {OTexts},
  title        = {Forecasting: Principles and Practice},
  year         = {2021}
}

@manual{hyndman2024,
  author       = {Hyndman, Rob J and Athanasopoulos, George and Bergmeir, Christoph and Caceres, Gabriel and Chhay, Leanne and O'Hara-Wild, Mitchell and Petropoulos, Fotios and Razbash, Slava and Wang, Earo and Yasmeen, Farah},
  note         = {R package version 9.0.0. \url{https://pkg.robjhyndman.com/forecast/}},
  title        = {forecast: Forecasting functions for time series and linear models},
  year         = {2026}
}

@article{kath2021conformal,
  author       = {Kath, Christopher and Ziel, Florian},
  journal      = {International Journal of Forecasting},
  number       = {2},
  pages        = {777--799},
  publisher    = {Elsevier},
  title        = {Conformal prediction interval estimation and applications to day-ahead and intraday power markets},
  volume       = {37},
  year         = {2021}
}

@article{KOO2015166,
  author       = {Bonsoo Koo and Myung Hwan Seo},
  doi          = {10.1016/j.jeconom.2015.03.046},
  issn         = {0304-4076},
  journal      = {Journal of Econometrics},
  number       = {1},
  pages        = {166--181},
  title        = {Structural-break models under mis-specification: Implications for forecasting},
  volume       = {188},
  year         = {2015}
}

@article{Koo17022026,
  author    = {Bonsoo Koo and Benjamin Wong and Ze-Yu Zhong},
  doi       = {10.1080/07350015.2025.2583205},
  journal   = {Journal of Business \& Economic Statistics},
  title     = {Disentangling Structural Breaks in Factor Models for Macroeconomic Data},
  publisher = {Taylor \& Francis},
  year      = {2026},
  note      = {In press}
}

@article{lei2018,
  author       = {Lei, Jing and G'Sell, Max and Rinaldo, Alessandro and Tibshirani, Ryan J. and Wasserman, Larry},
  journal      = {Journal of the American Statistical Association},
  number       = {523},
  pages        = {1094--1111},
  publisher    = {Taylor \& Francis},
  title        = {Distribution-Free Predictive Inference for Regression},
  volume       = {113},
  year         = {2018}
}

@article{Lipiecki2026,
  author       = {Arkadiusz Lipiecki and Kaja Bilińska and Nikolaos Kourentzes and Rafał Weron},
  doi          = {10.1016/j.ijforecast.2026.03.003},
  issn         = {0169-2070},
  journal      = {International Journal of Forecasting},
  number       = {4},
  pages        = {1215--1225},
  title        = {Stealing accuracy: Predicting day-ahead electricity prices with temporal hierarchy forecasting ({THieF})},
  volume       = {42},
  year         = {2026}
}

@article{montero2020fforma,
  author       = {Montero-Manso, Pablo and Athanasopoulos, George and Hyndman, Rob J and Talagala, Thiyanga S},
  journal      = {International Journal of Forecasting},
  number       = {1},
  pages        = {86--92},
  publisher    = {Elsevier},
  title        = {{FFORMA}: Feature-based forecast model averaging},
  volume       = {36},
  year         = {2020}
}

@inproceedings{papadopoulos2002,
  author       = {Papadopoulos, Harris and Proedrou, Kostas and Vovk, Volodya and Gammerman, Alex},
  booktitle    = {European Conference on Machine Learning},
  organization = {Springer},
  pages        = {345--356},
  title        = {Inductive confidence machines for regression},
  year         = {2002}
}

@article{Petropoulos2018,
  author       = {Fotios Petropoulos and Rob J. Hyndman and Christoph Bergmeir},
  doi          = {10.1016/j.ejor.2018.01.045},
  issn         = {0377-2217},
  journal      = {European Journal of Operational Research},
  number       = {2},
  pages        = {545--554},
  title        = {Exploring the sources of uncertainty: Why does bagging for time series forecasting work?},
  volume       = {268},
  year         = {2018}
}

@article{petropoulos2022,
  author       = {Petropoulos, Fotios and Apiletti, Daniele and Assimakopoulos, Vassilios and Babai, Mohamed Zied and Barrow, Devon K and Taieb, Souhaib Ben and Bergmeir, Christoph and Bessa, Ricardo J and Bijak, Jakub and Boylan, John E and others},
  journal      = {International Journal of Forecasting},
  number       = {3},
  pages        = {705--871},
  publisher    = {Elsevier},
  title        = {Forecasting: theory and practice},
  volume       = {38},
  year         = {2022}
}

@article{schlembach2025,
  author       = {Schlembach, Filip and Smirnov, Evgueni and Koprinska, Irena and Winands, Mark H. M.},
  journal      = {Machine Learning},
  number       = {7},
  pages        = {165},
  publisher    = {Springer},
  title        = {Conformal multistep-ahead multivariate time-series forecasting},
  volume       = {114},
  year         = {2025}
}

@article{shafer2008,
  author       = {Shafer, Glenn and Vovk, Vladimir},
  journal      = {Journal of Machine Learning Research},
  number       = {12},
  pages        = {371--421},
  title        = {A Tutorial on Conformal Prediction},
  volume       = {9},
  year         = {2008}
}

@article{shang2026conformal,
  title        = {Conformal prediction for functional time series: application to age-specific mortality rates},
  author       = {Shang, Han Lin},
  journal      = {Journal of Population Research},
  volume       = {43},
  number       = {2},
  pages        = {14},
  year         = {2026},
  publisher    = {Springer}
}

@article{stankeviciute2021,
  author       = {Stankeviciute, Kamile and Alaa, Ahmed M. and van der Schaar, Mihaela},
  journal      = {Advances in Neural Information Processing Systems},
  pages        = {6216--6228},
  title        = {Conformal time-series forecasting},
  volume       = {34},
  year         = {2021}
}

@article{sun2025,
  author       = {Sun, Sophia and Yu, Rose},
  journal      = {Advances in Neural Information Processing Systems},
  pages        = {12243--12272},
  title        = {Conformal prediction for time-series forecasting with change points},
  volume       = {38},
  year         = {2025}
}

@inproceedings{tibshirani2019,
  author       = {Tibshirani, Ryan J. and Barber, Rina Foygel and Cand\`es, Emmanuel J. and Ramdas, Aaditya},
  booktitle    = {Advances in Neural Information Processing Systems},
  pages        = {2530--2540},
  title        = {Conformal Prediction under Covariate Shift},
  volume       = {32},
  year         = {2019}
}

@book{vovk2005,
  author       = {Vovk, Vladimir and Gammerman, Alexander and Shafer, Glenn},
  doi          = {10.1007/b106715},
  publisher    = {Springer-Verlag},
  title        = {Algorithmic Learning in a Random World},
  year         = {2005}
}

@article{wang2024cpts,
  author       = {Wang, Xiaoqian and Hyndman, Rob J},
  journal      = {arXiv preprint arXiv:2410.13115},
  title        = {Online conformal inference for multi-step time series forecasting},
  year         = {2024}
}

@article{Weron2014,
  author       = {Rafał Weron},
  doi          = {10.1016/j.ijforecast.2014.08.008},
  issn         = {0169-2070},
  journal      = {International Journal of Forecasting},
  number       = {4},
  pages        = {1030--1081},
  title        = {Electricity price forecasting: A review of the state-of-the-art with a look into the future},
  volume       = {30},
  year         = {2014}
}

@article{winkler1972,
  author       = {Winkler, Robert L},
  journal      = {Journal of the American Statistical Association},
  number       = {337},
  pages        = {187--191},
  publisher    = {Taylor \& Francis},
  title        = {A decision-theoretic approach to interval estimation},
  volume       = {67},
  year         = {1972}
}

@article{wintenberger2017,
  author       = {Wintenberger, Olivier},
  journal      = {Machine Learning},
  number       = {1},
  pages        = {119--141},
  publisher    = {Springer},
  title        = {Optimal learning with Bernstein online aggregation},
  volume       = {106},
  year         = {2017}
}

@inproceedings{xu2021,
  author       = {Xu, Chen and Xie, Yao},
  booktitle    = {International Conference on Machine Learning},
  organization = {PMLR},
  pages        = {11559--11569},
  title        = {Conformal Prediction Interval for Dynamic Time-Series},
  volume       = {139},
  year         = {2021}
}

@inproceedings{xu2023,
  author       = {Xu, Chen and Xie, Yao},
  booktitle    = {International Conference on Machine Learning},
  organization = {PMLR},
  pages        = {38707--38727},
  title        = {Sequential predictive conformal inference for time series},
  year         = {2023}
}

@article{yang2024ts,
  author       = {Yang, Zitong and Cand{\`e}s, Emmanuel and Lei, Lihua},
  journal      = {arXiv preprint arXiv:2402.05203},
  title        = {Bellman Conformal Inference: Calibrating Prediction Intervals For Time Series},
  year         = {2024}
}

@inproceedings{zaffran2022,
  author       = {Zaffran, Margaux and F{\'e}ron, Olivier and Goude, Yannig and Josse, Julie and Dieuleveut, Aymeric},
  booktitle    = {International Conference on Machine Learning},
  organization = {PMLR},
  pages        = {25834--25866},
  title        = {Adaptive conformal predictions for time series},
  year         = {2022}
}

\clearpage
\appendix
\crefalias{section}{appendixref}
\crefalias{subsection}{appendixref}
\crefalias{subsubsection}{appendixref}
\section{Proofs for the theoretical results}
\label{app:theory-proofs}

\subsection{Fixed-share auxiliary results}
\label{app:fixed-share-lemmas}

Define the exponential normalizer
\begin{equation*}
    Z_t
    :=\sum_{m=1}^M
       w_t^{(m)}\exp\{-\xi\loss_t^{(m)}\}.
\end{equation*}
Then
\begin{equation*}
    \widetilde w_{t+1}^{(m)}
    =
    \frac{w_t^{(m)}\exp\{-\xi\loss_t^{(m)}\}}{Z_t}.
\end{equation*}

\begin{lemma}[One-step exponential-weights inequality]
\label{lem:normalizer-upper}
Under \Cref{ass:finite-meta-loss},
\begin{equation}
    \log Z_t
    \le -\xi\bar\loss_t+\xi^2\bar v_t.
    \label{eq:normalizer-upper}
\end{equation}
\end{lemma}

\begin{proof}
For $x\ge0$, $\exp(-x)\le1-x+x^2$. Applying this inequality with
$x=\xi\loss_t^{(m)}$, multiplying by $w_t^{(m)}$, and summing over
$m$ gives
\begin{equation*}
    Z_t\le1-\xi\bar\loss_t+\xi^2\bar v_t.
\end{equation*}
The right-hand side is positive because it is at least $Z_t>0$.
Monotonicity of the logarithm and $\log(1+u)\le u$ yield
\Cref{eq:normalizer-upper}.
\end{proof}

\begin{lemma}[Expert recoverability]
\label{lem:fixed-share-floor}
Under \Cref{ass:finite-meta-loss}, for every $t\ge1$ and every expert $m$,
\begin{equation*}
    w_t^{(m)}\ge\frac{\sigma}{M}.
\end{equation*}
\end{lemma}

\begin{proof}
At initialization, $w_1^{(m)}=1/M\ge\sigma/M$. For $t\ge1$,
\Cref{eq:fixed-share-weight-update} and
$\widetilde w_{t+1}^{(m)}\ge0$ imply
$w_{t+1}^{(m)}\ge\sigma/M$.
\end{proof}

\begin{lemma}[Local comparison with a fixed expert]
\label{lem:normalizer-lower}
Under \Cref{ass:finite-meta-loss}, for every interval $\mathcal I=[r,s]$ and every expert $m$,
\begin{equation}
    \sum_{t=r}^s\log Z_t
    \ge
    \log\frac{\sigma}{M}
    +L\log(1-\sigma)
    -\xi\sum_{t=r}^s\loss_t^{(m)}.
    \label{eq:normalizer-lower}
\end{equation}
\end{lemma}

\begin{proof}
The fixed-share update satisfies
\begin{align*}
    w_{t+1}^{(m)}
    &\ge(1-\sigma)\widetilde w_{t+1}^{(m)}\\
    &=(1-\sigma)
      \frac{w_t^{(m)}\exp\{-\xi\loss_t^{(m)}\}}{Z_t}.
\end{align*}
Consequently,
\begin{equation*}
    Z_t\ge(1-\sigma)
    \frac{w_t^{(m)}}{w_{t+1}^{(m)}}
    \exp\{-\xi\loss_t^{(m)}\}.
\end{equation*}
Multiplying from $t=r$ to $s$ telescopes the weight ratios and gives
\begin{equation*}
    \prod_{t=r}^s Z_t
    \ge(1-\sigma)^L
       \frac{w_r^{(m)}}{w_{s+1}^{(m)}}
       \exp\!\left\{-\xi\sum_{t=r}^s\loss_t^{(m)}\right\}.
\end{equation*}
Because $w_{s+1}^{(m)}\le1$ and, by
\Cref{lem:fixed-share-floor}, $w_r^{(m)}\ge\sigma/M$, taking logarithms
proves \Cref{eq:normalizer-lower}.
\end{proof}

\subsection{Proofs of the regret and parameter results}
\label{app:regret-parameter-proofs}

\begin{proof}[Proof of \Cref{thm:local-fixed-expert-regret}]
Summing \Cref{eq:normalizer-upper} over $\mathcal I$ and combining it with
\Cref{eq:normalizer-lower} yields
\begin{equation*}
    \log\frac{\sigma}{M}+L\log(1-\sigma)
    -\xi\sum_{t=r}^s\loss_t^{(m)}
    \le
    -\xi\sum_{t=r}^s\bar\loss_t
    +\xi^2\sum_{t=r}^s\bar v_t.
\end{equation*}
Rearranging and dividing by $\xi$ proves
\Cref{eq:local-regret-exact}. Since
$-\log(1-\sigma)\le2\sigma$ for $0<\sigma\le1/2$,
\Cref{eq:local-regret-simplified} follows.
\end{proof}

\begin{proof}[Proof of \Cref{cor:share-rule}]
For fixed $\xi$, the relevant term is
$g(\sigma)=\log(M/\sigma)+2\sigma L$. It is strictly convex because
$g''(\sigma)=1/\sigma^2>0$, and
$g'(\sigma)=-1/\sigma+2L$ vanishes at $\sigma=1/(2L)$.
For $L\ge1$, this value satisfies $0<\sigma\le1/2$.
\end{proof}

\paragraph{Reference calculation of the learning rate}
After substituting $\sigma=1/(2L)$ into \Cref{eq:local-regret-simplified}, hold $S_\mathcal I>0$ fixed and consider the sum
\begin{equation*}
    \xi S_\mathcal I+\frac{\log(2ML)+1}{\xi}.
\end{equation*}
As a function of $\xi>0$ with $S_\mathcal I$ fixed, its first and second derivatives are
\begin{equation*}
    S_\mathcal I-\frac{\log(2ML)+1}{\xi^2},
    \qquad
    \frac{2\{\log(2ML)+1\}}{\xi^3}>0,
\end{equation*}
respectively. Its unique minimum occurs at \Cref{eq:xi-realized-optimum}. In an actual run, the learning rate affects the realized weights and aggregate coverage feedback, and hence $S_\mathcal I$. The minimization therefore treats $S_\mathcal I$ as fixed. Replacing $V_\mathcal I$ in \Cref{eq:xi-realized-optimum} with the prespecified scale $V$ gives the learning-rate rule in \Cref{eq:aggregation-rules}.

\subsection{Comparison with an external benchmark sequence}
\label{app:target-sequence-comparison}

The fixed-expert comparison extends naturally to an external benchmark sequence
when the library contains a fixed expert that approximates that benchmark well
over the interval of interest. We formalize this requirement as follows.

\begin{assumption}[Expert-library approximation]
\label{ass:grid-approximation}
For a benchmark sequence with finite losses $\{\loss_t^\star\}_{t\in \mathcal I}$, suppose that the library contains an expert $m_\mathcal I$ satisfying
\begin{equation}
    \sum_{t=r}^s\loss_t^{(m_\mathcal I)}
    -\sum_{t=r}^s\loss_t^\star
    \le\Delta_{\mathrm{grid}}(\mathcal I),
    \qquad 0\le\Delta_{\mathrm{grid}}(\mathcal I)<\infty.
    \label{eq:grid-approximation}
\end{equation}
\end{assumption}

\begin{corollary}[Local comparison with a target sequence]
\label{thm:target-sequence-regret}
Under \Cref{ass:finite-meta-loss,ass:grid-approximation},
\begin{equation}
    \sum_{t=r}^s\bar\loss_t
    -\sum_{t=r}^s\loss_t^\star
    \le
    \xi\sum_{t=r}^s\bar v_t
    +\frac{\log(M/\sigma)+2\sigma L}{\xi}
    +\Delta_{\mathrm{grid}}(\mathcal I).
    \label{eq:target-sequence-regret}
\end{equation}
\end{corollary}

The corollary separates aggregation regret from library approximation error. Since $\Delta_{\mathrm{grid}}(\mathcal I)$ is not controlled here as a function of library resolution, the result is an error decomposition rather than a quantitative approximation guarantee.

\begin{proof}[Proof of \Cref{thm:target-sequence-regret}]
Apply \Cref{thm:local-fixed-expert-regret} with $m=m_\mathcal I$ and add
\Cref{eq:grid-approximation}. The intermediate loss of $m_\mathcal I$ cancels,
yielding \Cref{eq:target-sequence-regret}.
\end{proof}

\subsection{An idealized reference-distribution derivation of the second-moment scale}
\label{app:reference-v-derivation}

To motivate a prespecified second-moment scale for the coverage--efficiency loss, we consider an idealized population reference problem. Let $S\ge0$ have a continuous distribution function $F$ and density $f$, with $(1-\alpha)$-quantile $0<Q=F^{-1}(1-\alpha)<\infty$. Assume that $F$ is differentiable at $Q$, with derivative $f(Q)>0$. Consider the population width minimization problem
\begin{equation*}
    \min_{q\ge0}\frac{2q}{B_q}
    \quad\text{subject to}\quad
    \Pr(S>q)\le\alpha.
\end{equation*}
At the calibrated solution $q=Q$, with the ideal normalization $B_q=2Q$,
the normalized width equals one. The Lagrangian is
\begin{equation*}
    \mathcal L(q,\lambda)
    =\frac{2q}{B_q}
      +\lambda\{\Pr(S>q)-\alpha\}.
\end{equation*}
Since $Q>0$, stationarity with respect to $q$ at $Q$ gives
\begin{equation*}
    \lambda^\star
    =\frac{2}{B_q f(Q)}
    =\frac{1}{Qf(Q)}.
\end{equation*}
We use the standardized reference score $U=S/Q$ and take
$U\sim\operatorname{Unif}(0,1/(1-\alpha))$. This reference law satisfies
$\Pr(U\le1)=1-\alpha$ and has $f_U(1)=1-\alpha$, so
\begin{equation*}
    \lambda^\star=\frac{1}{f_U(1)}
    =\frac{1}{1-\alpha}.
\end{equation*}

For a reference calculation over a single update, set every expert threshold, and hence the aggregate threshold, equal to $Q$, and set the multiplier before the update equal to $\lambda^\star$. All expert miscoverage indicators then coincide with the aggregate indicator, with miscoverage probability $\alpha$.

As in \Cref{eq:lambda-update-main,eq:binary-meta-loss}, the current aggregate error updates the multiplier before the expert loss is evaluated. With $B_q=2Q$, a covered observation therefore has loss $1$, while a miscovered observation has loss
\begin{equation*}
    1+\lambda^\star+\nu(1-\alpha).
\end{equation*}
Since these events have probabilities $1-\alpha$ and $\alpha$,
respectively, the resulting reference second moment is
\begin{equation*}
    V_{\mathrm{ref}}
    =(1-\alpha)
     +\alpha\left[
       1+\frac{1}{1-\alpha}+\nu(1-\alpha)
     \right]^2.
\end{equation*}
For the experiments, $\alpha=0.1$ and $\nu=0.02$, yielding $V_{\mathrm{ref}}\approx1.3533$; before evaluation, we set $V$ in \Cref{eq:aggregation-rules} to this rounded value. The calculation describes a population reference state and a single update. It makes no claim that the online multiplier converges to $\lambda^\star$ or that $V_{\mathrm{ref}}$ is a stationary second moment of the algorithm. The resulting $V$ serves as a reference tuning scale, with no claim that it estimates or bounds the realized $V_\mathcal I$ or yields optimal online tuning.

\subsection{Proof of the pathwise aggregate miscoverage bound}
\label{app:multiplier-proof}

\begin{proof}[Proof of \Cref{prop:lambda-certificate}]
Projection onto the nonnegative half-line implies
\begin{equation*}
    \lambda_t
    \ge\lambda_{t-1}+\nu(err_t-\alpha).
\end{equation*}
Hence
$err_t-\alpha\le(\lambda_t-\lambda_{t-1})/\nu$.
Summing from $t=1$ to $T$ telescopes and gives \Cref{eq:lambda-bound}. Taking the upper limit along any realized path for which $\lambda_T/T\to0$ proves \Cref{eq:lambda-asymptotic-upper}.
\end{proof}

\subsection{Proof of stochastic multiplier sublinearity}
\label{app:multiplier-stability-proof}

\begin{proof}[Proof of \Cref{prop:multiplier-mean-square}]
Because $[x]_+^2\le x^2$, the multiplier update satisfies
\begin{align*}
    \lambda_t^2
    &\le
    \left\{\lambda_{t-1}+\nu(err_t-\alpha)\right\}^2
    \notag\\
    &=
    \lambda_{t-1}^2
    +2\nu\lambda_{t-1}(err_t-\alpha)
    +\nu^2(err_t-\alpha)^2.
\end{align*}
Since $err_t\in\{0,1\}$ and $0<\alpha<1$, we have $(err_t-\alpha)^2\le1$. Using the $\mathcal F_{t-1}$-measurability of $\lambda_{t-1}$ gives
\begin{equation*}
    \E[\lambda_t^2\mid\mathcal F_{t-1}]
    \le \lambda_{t-1}^2
    +2\nu\lambda_{t-1}\E[err_t-\alpha\mid\mathcal F_{t-1}]
    +\nu^2.
\end{equation*}
We bound the conditional drift term in two cases. If $\lambda_{t-1}<\Lambda$, then
\begin{equation*}
    2\nu\lambda_{t-1}
    \E[err_t-\alpha\mid\mathcal F_{t-1}]
    \le 2\nu\Lambda(1-\alpha).
\end{equation*}
If $\lambda_{t-1}\ge\Lambda$, then
\Cref{eq:high-multiplier-drift} gives
\begin{equation*}
    2\nu\lambda_{t-1}
    \E[err_t-\alpha\mid\mathcal F_{t-1}]
    \le -2\nu\delta\lambda_{t-1}
    \le0.
\end{equation*}
Therefore, in both cases,
\begin{equation*}
    \E[\lambda_t^2\mid\mathcal F_{t-1}]
    \le\lambda_{t-1}^2+C_\lambda.
\end{equation*}
Taking expectations and iterating from $t=1$ to $T$ yields
\begin{equation*}
    \E[\lambda_T^2]
    \le\lambda_0^2+C_\lambda T.
\end{equation*}
Dividing by $T^2$ proves
\Cref{eq:lambda-l2-sublinear}. Mean-square convergence implies convergence in
probability.

To establish almost-sure sublinearity, apply Markov's inequality along the subsequence $T=k^2$. For every $\varepsilon>0$,
\begin{equation*}
    \sum_{k=1}^{\infty}
    \Pr\!\left(\frac{\lambda_{k^2}}{k^2}>\varepsilon\right)
    \le \frac{1}{\varepsilon^2}
    \sum_{k=1}^{\infty}
    \left(\frac{\lambda_0^2}{k^4}+\frac{C_\lambda}{k^2}\right)
    <\infty.
\end{equation*}
The first Borel--Cantelli lemma, applied to $\varepsilon=1/j$ for positive integers $j$, yields $\lambda_{k^2}/k^2\to0$ almost surely. The multiplier update also gives $|\lambda_t-\lambda_{t-1}|\le\nu$. Hence, for $k^2\le T<(k+1)^2$,
\begin{equation*}
    0\le\frac{\lambda_T}{T}
    \le\frac{\lambda_{k^2}+\nu(2k+1)}{k^2}
    \longrightarrow0
    \qquad\text{almost surely}.
\end{equation*}
This proves the almost-sure conclusion stated after the proposition.
\end{proof}

\subsection{Proof of expected aggregate miscoverage}
\label{app:expected-miscoverage-proof}

\begin{proof}[Proof of \Cref{cor:expected-miscoverage}]
Taking expectations in \Cref{eq:lambda-bound}, using $\lambda_0\ge0$, and
then applying the Cauchy--Schwarz inequality gives
\begin{align*}
    \E\!\left[\frac{1}{T}\sum_{t=1}^T err_t\right]
    &\le
    \alpha+\frac{\E[\lambda_T]-\lambda_0}{\nu T}\\
    &\le
    \alpha+\frac{\E[\lambda_T]}{\nu T}\\
    &\le
    \alpha+\frac{\sqrt{\E[\lambda_T^2]}}{\nu T}.
\end{align*}
Substituting \Cref{eq:lambda-second-moment-bound} proves
\Cref{eq:expected-miscoverage-rate}.
\end{proof}

\section{PID expert libraries}
\label{app:expert-pools}

\subsection{Reference saturation constant}

We construct the PID expert libraries around the reference Conformal PID configuration. The reference saturation constant is determined by the logarithmic rule
\begin{equation}
    C_{\mathrm{sat}}^\star(T_{\mathrm{ref}})
    =
    \frac{2}{\pi}
    \left\{
        \left\lceil
            \delta_{\mathrm{sat}}\log T_{\mathrm{ref}}
        \right\rceil
        -\frac{1}{\log T_{\mathrm{ref}}}
    \right\},
    \qquad
    \delta_{\mathrm{sat}}=0.01.
    \label{eq:csat-reference-rule}
\end{equation}
The logarithm in \Cref{eq:csat-reference-rule} is natural,
$\lceil\cdot\rceil$ denotes the ceiling function, and
$T_{\mathrm{ref}}>1$ is the horizon used to construct the reference
controller. The subscript in $\delta_{\mathrm{sat}}$ distinguishes this
numerical saturation constant.

\subsection{Nonlinear simulation library}

For the nonlinear simulation, we set $T_{\mathrm{ref}}=1000$, giving
\begin{equation*}
    C_{\mathrm{sat}}^\star(1000)
    =0.5444596\ldots.
\end{equation*}
The library comprises all $M=7^3=343$ combinations of the following parameter values:
\begin{align*}
    \eta
    &\in
    \{0.025,0.05,0.075,0.1,0.125,0.15,0.2\},
    \nonumber\\
    K_I
    &\in
    \{0.1,0.25,0.4,0.5,0.75,1.0,1.5\},
    \nonumber\\
    C_{\mathrm{sat}}
    &\in
    \left\{
        0.2,0.3,0.4,
        C_{\mathrm{sat}}^\star(1000),
        0.7,0.85,1.0
    \right\}.
\end{align*}
The reference saturation constant is reported as $0.544$ when rounded.

\subsection{Distribution-shift simulation library}

All three distribution-shift scenarios use the same PID expert library, constructed with reference sequence length $T_{\mathrm{ref}}=6000$. The resulting reference saturation constant is
\begin{equation*}
    C_{\mathrm{sat}}^\star(6000)
    =0.5634410\ldots.
\end{equation*}
The library consists of all $M=7^3=343$ combinations of the following parameter values:
\begin{align*}
    \eta
    &\in
    \{0.025,0.05,0.075,0.1,0.125,0.15,0.2\},
    \nonumber\\
    K_I
    &\in
    \{0.5,1.0,1.5,2.0,2.5,3.0,4.0\},
    \nonumber\\
    C_{\mathrm{sat}}
    &\in
    \{0.2,0.3,0.4,0.56,0.7,0.85,1.0\}.
\end{align*}
The grid uses the rounded reference value $0.56$, whereas the Conformal PID benchmark uses $C_{\mathrm{sat}}^\star(6000)$ without rounding, together with $\eta=0.1$ and $K_I=2$.

\subsection{Real-data expert libraries}

For each forecasting model, let $T_s$ be the number of nonconformity scores available for calibration and $L_0$ the length of the initial score window. We define the reference scales
\begin{equation*}
    B'
    =
    \max_{1\le t\le \min(L_0,T_s)} |s_t|,
    \qquad
    C_{\mathrm{sat}}^\star
    =
    C_{\mathrm{sat}}^\star(T_s).
\end{equation*}
In both applications, $T_s\ge L_0$ and the scores are absolute forecast errors, so $B'=\max_{1\le t\le L_0}s_t$.

The Ordos energy consumption and Capital Bikeshare applications each use all $M=9^3=729$ parameter combinations below:
\begin{align*}
    \eta
    &\in
    \{0.02,0.025,0.05,0.075,0.1,
      0.125,0.15,0.2,0.25\},
    \nonumber\\
    K_I
    &\in
    B'\times
    \{0.125,0.25,0.5,0.75,1,
      1.25,1.5,2,4\},
    \nonumber\\
    C_{\mathrm{sat}}
    &\in
    C_{\mathrm{sat}}^\star\times
    \{0.25,0.35,0.5,0.75,1,
      1.25,1.5,2,2.5\}.
\end{align*}
The Conformal PID benchmark uses the same reference scales:
\begin{equation*}
    \eta=0.1,
    \qquad
    K_I=B',
    \qquad
    C_{\mathrm{sat}}=C_{\mathrm{sat}}^\star.
\end{equation*}

\section{DtACI and AgACI configurations}
\label{app:aci-benchmarks}

The ACI benchmarks aggregate adaptive miscoverage levels before converting them to thresholds. AgACI here denotes an adaptation to calibration levels; the original method aggregates interval endpoints separately \citep{zaffran2022}.

DtACI and AgACI share a library of nine ACI experts, distinguished by their step sizes:
\begin{equation*}
    \gamma
    \in
    \{0.001,0.002,0.004,0.008,0.016,
      0.032,0.064,0.128,0.256\},
    \qquad M_{\mathrm{ACI}}=9.
\end{equation*}

For DtACI, the share parameter and learning rate are set to
\begin{equation*}
    \sigma_{\mathrm{DtACI}}
    =\frac{1}{2L},
    \qquad
    \eta_{\mathrm{DtACI}}
    =
    \sqrt{
        \frac{\log(2M_{\mathrm{ACI}}L)+1}
             {L V_{\mathrm{pin}}}
    },
    \qquad
    V_{\mathrm{pin}}
    =
    \frac{\alpha^2(1-\alpha)^2}{3},
\end{equation*}
where $L$ is the local adaptation window length. AgACI uses adaptive learning rates for individual experts, with stabilization constant $\varepsilon_{\mathrm{AgACI}}=0.001$.

For the nonlinear simulation and real-data applications, both methods use the rolling calibration window
\begin{equation*}
    \mathcal W_t
    =
    \{\max(1,t-L),\ldots,t-1\}.
\end{equation*}
Using only preceding scores, the empirical upper-tail probability of the observed $s_t$ is
\begin{equation*}
    \beta_t
    =
    \frac{
        1+\sum_{r\in\mathcal W_t}\ind\{s_r\ge s_t\}
    }{
        |\mathcal W_t|+1
    }.
\end{equation*}

The aggregate level $\overline\alpha_t$, chosen before observing $s_t$, determines the reported threshold:
\begin{equation*}
    q_t
    =
    \widehat Q_{1-\overline\alpha_t}
    \bigl(\{s_r:r\in\mathcal W_t\}\bigr),
\end{equation*}
where $\widehat Q$ is the higher empirical quantile. We set $q_t=0$ when $\mathcal W_t$ is empty.

For the Gaussian distribution-shift simulations, DtACI and AgACI use a fixed standard normal reference to relate the observed nonconformity score to a tail probability and the adaptive miscoverage level to an interval threshold:
\begin{equation*}
    \beta_t
    =
    2\{1-\Phi(|Y_t|)\},
    \qquad
    q_t
    =
    \Phi^{-1}
    \left(1-\frac{\overline\alpha_t}{2}\right),
\end{equation*}
where $\Phi$ is the standard normal distribution function and $\overline\alpha_t$ is determined before observing $Y_t$. The reference distribution remains fixed across scenarios; adaptation operates through $\overline\alpha_t$, without using the time-varying mean $\mu_t$.

\end{document}